\documentclass[11pt]{article}

\usepackage[lined,boxed,ruled,norelsize,algo2e,linesnumbered]{algorithm2e}
\usepackage{amssymb,amsthm,amsmath}
\usepackage{libertine}
\usepackage[libertine]{newtxmath}
\usepackage{verbatim}


\usepackage{subfig}
\usepackage[usenames,dvipsnames]{xcolor}
\usepackage[colorlinks,citecolor=blue,linkcolor=BrickRed]{hyperref}\usepackage[colorlinks,citecolor=blue,linkcolor=BrickRed]{hyperref}
\usepackage{fullpage}
\usepackage{graphicx}
\usepackage{xspace}
\usepackage{enumitem}
\usepackage{thmtools,thm-restate}
\usepackage{wrapfig}
\usepackage{algorithm}
\usepackage{algpseudocode}

\graphicspath{{./fig/}}
\newtheorem{theorem}{Theorem}

\newtheorem{lemma}[theorem]{Lemma}

\newtheorem{corollary}[theorem]{Corollary}

\newtheorem{remark}[theorem]{Remark}
\newtheorem{claim}[theorem]{Claim}

\newtheorem{definition}[theorem]{Definition}
\newif\ifFULL
\FULLtrue

\newcounter{note}[section]

\newcommand{\daogao}[1]{\refstepcounter{note}$\ll${\bf Daogao~\thenote:}
  {\sf \color{green}  #1}$\gg$\marginpar{\tiny\bf DL~\thenote}}

\newcommand{\val}{\mathsf{val}}

\renewcommand{\index}{\Gamma}

\newcommand{\Z}{C_{mv}}

\newcommand{\calF}{\mathcal{F}}

\newcommand{\R}{R}
\newcommand{\ER}{\bar{R}}

\newcommand{\EC}{\bar{C}}

\newcommand{\SC}{C_{sw}} %switching cost
\newcommand{\ESC}{\bar{C}_{sw}} %expected switch cost
\newcommand{\MC}{C_{mv}} %mov cost
\newcommand{\EMC}{\bar{C}_{mv}} %expected mov cost
\newcommand{\TC}{C_{tot}} %total cost
\newcommand{\ETC}{\bar{C}_{tot}} %expected total cost
\newcommand{\CSC}{{C}_{sl}} %selected cost
\newcommand{\TSPCOST}{C_{ktsp}}
\newcommand{\FMC}{C_{mv}^F}
\newcommand{\Ftotalcost}{C_{tot}^F}

\renewcommand{\P}{\mathbb{P}} %strategy

\newcommand{\Pipref}{\Pi_{pref}}

\newcommand{\D}{\ensuremath{C_{sl}}\xspace}
\newcommand{\p}{\ensuremath{\mathsf{P}}}

\newcommand{\E}{\mathbb{E}}
\newcommand{\F}{\mathcal{F}}
\newcommand{\M}{\mathcal{M}}
\newcommand{\T}{\mathcal{T}}
\newcommand{\PG}{\alpha_{\mathrm{ktsp}}}
\renewcommand{\S}{\ensuremath{\mathcal{S}}\xspace}
\newcommand{\GS}{\ensuremath{\mathbb{S}}\xspace}

\newcommand{\CGame}{\mathcal{G}}

\newcommand{\dist}{d}

\newcommand{\OPT}{\ensuremath{\mathbb{O}}\xspace}
\newcommand{\FOPT}{\ensuremath{\mathbb{O}_{fair}}\xspace}
\newcommand{\SCOPT}{\ensuremath{\mathbb{O}_{ktsp}}\xspace}

\newcommand{\ALG}{\ensuremath{\mathbb{ALG}}\xspace}
\newcommand{\ALGindex}{\ensuremath{\mathsf{ALG_{index}}}\xspace}

\newcommand{\ind}{\ensuremath{\mathbf{1}}}
\newcommand{\eat}[1]{}
\newcommand{\hide}[1]{{\Large \color{red} Contents here are hidden! To reveal contents, remove this command.}}

\newcommand{\G}{\ensuremath{\mathcal{G}}}

\newcommand{\ALGBO}{\ensuremath{\ALG_{\textsf{Bicrit-Orient}}}\xspace}

\newcommand{\METAfindone}{\ensuremath{\mathsf{BudgetMG\text{-}Unit}}\xspace}

\newcommand{\ALGSfindK}{\ensuremath{\mathsf{MG\text{-}Metric}}\xspace}
\newcommand{\METASfindK}{\ensuremath{\mathsf{BudgetMG\text{-}Metric}}\xspace}

\newcommand{\ALGmeta}{\ensuremath{\mathsf{BudgetMG}}\xspace}
\newcommand{\ALGgeneral}{\ensuremath{\mathsf{Algo\text{-}MG}}\xspace}

\newcommand{\depot}{\ensuremath{\mathsf{root}}\xspace}

\newcommand{\KTSP}{\ensuremath{\textsf{$k$-TSP}}\xspace}
\newcommand{\ALGkTSP}{\ensuremath{\ALG_{\KTSP}}\xspace}

\newcommand{\StochKTSP}{\ensuremath{\textsf{Stochastic Reward $k$-TSP}}\xspace}
\newcommand{\StochKCost}{\ensuremath{\textsf{Stochastic- $k$-TSP}}\xspace}

\newcommand{\poly}{\ensuremath{\mathrm{poly}}\xspace}

\newcommand{\ALGStochKCost}{\ensuremath{\mathsf{\ALG_{\textsf{ktsp}}}}\xspace}

\newcommand{\per}{\mathrm{per}}

\newcommand{\IGNORE}[1]{}

\newcommand{\FindOne}{\textsf{MG-Unit}\xspace}
\newcommand{\SFindK}{\textsf{MG-Metric}\xspace}
\newcommand{\SFindKF}{\textsf{MG-Metric-Fair}\xspace}

\newcommand{\MS}{Markov system\xspace}

\newcommand{\ALGGITTINS}{\mathbb{GT}}

\newenvironment{proofof}[1]{\smallskip\noindent{\bf Proof of #1.}}%
        {\hspace*{\fill}$\Box$\par}
        
\def \ifArxiv
{
\author{
Jian Li
\thanks{Institute for Interdisciplinary Information Sciences, Tsinghua University. Email:\texttt{lijian83@mail.tsinghua.edu.cn.}}
\and
Daogao Liu
\thanks{Paul G. Allen School of Computer Science \& Engineering, University of Washington. %Supported by NSF awards CCF-1749609, CCF-1740551, DMS-1839116, DMS-2023166, Microsoft Research Faculty Fellowship, Sloan Research Fellowship, Packard Fellowships. 
Email:  \texttt{dgliu@cs.washington.edu.}}
}
}
\title{Multi-token Markov Game with Switching Costs}

\ifdefined\ifArxiv
\ifArxiv
\else
\author{Anonymous Author(s)}
\fi
\date{}

\begin{document}
\maketitle
\begin{abstract}
We study a general Markov game with metric switching costs: 
in each round, the player adaptively chooses one of several Markov chains to advance with the objective of minimizing the expected cost for at least $k$ chains to reach their target states. If the player decides to play a different chain, an additional switching cost is incurred. The special case in which there is no switching cost was solved optimally by Dumitriu, Tetali and Winkler~\cite{DTW03} by a variant of 
the celebrated Gittins Index for the classical multi-armed bandit (MAB) problem with Markovian rewards \cite{Git74,Git79}.
However, for Markovian multi-armed bandit with nontrivial switching cost, even if the switching cost is a constant, the classic paper by Banks and Sundaram \cite{BS94} showed that no index strategy can be optimal.
\footnote{
Their proof is for the discounted version of MAB, but can be extended to our setting. See Appendix~\ref{sec:no_opt_index} for the details.
}

In this paper, we complement their result and show there is a simple index strategy that achieves a constant approximation factor if the switching cost is constant and $k=1$. 
To the best of our knowledge, this index strategy is the first strategy that achieves a constant approximation factor for a general Markovian MAB variant with switching costs. For the general metric, we propose a more involved
constant-factor approximation algorithm, via a nontrivial reduction to the stochastic $k$-TSP problem, in which a Markov chain is approximated by a random variable.
Our analysis makes extensive use of various interesting properties of the Gittins index.
\end{abstract}
\thispagestyle{empty}
\newpage

%\clearpage
\setcounter{page}{1}

\section{Introduction}
The Markovian multi-armed bandit (MAB) problem is one of the most important and well studied sequential decision problem. In this problem, at each time step,
the agent knows the state of each chain and must choose to play one of $n$ available Markov chains. The agent pays a certain cost (or receives a  payoff) depending on the current state, and the chosen chain advances to the next state (according the Markovian transition rules).
The goal of the agent is to optimize the expected cost (or payoff) by choosing the right sequence of actions. 
The infinite horizon discounted version of the problem was solved optimally by the celebrated Gittins Index theorem, first proved by Gittins and Jones\cite{Git74}. In particular, they show that the Markov chains can be ``indexed'' separately, and the optimal strategy is simply choosing the chain with the smallest (or largest) index. Since then, Gittins Index has been studied and extent in a variety of ways
(see \cite{GGW11}).

A major extension to MAB is the inclusion of switching costs, that is switching to a different chain incurs a nontrivial cost
\cite{BS94,AT96,KSU08,KLM17,CGT+20}.
This extension has found many applications in job search and labor mobility \cite{J78,V80,M80,M84,W84,J84,KW11}, industrial policy \cite{PT95,K04}, optimal search \cite{Wei79,BB88,BGO92,S99}, experiment and learning \cite{R74,M84,K03,AKW04} and game theory \cite{S97,BV06}.
The most natural problem for MAB in the presence of switching costs
is to examine the extend to which the Gittins-Jones theorem remains valid,
i.e., whether there is a suitably defined index strategy that is optimal. This problem was first studied in the classic paper by Banks and Sundaram \cite{BS94}, who showed that there is no index strategy that is optimal, even if the switching cost is a given nonzero constant.
Motivated by this work, several authors \cite{BGO92,AT96,VOP00,KL00,BV01,J01,DH03} attempted to (partially) characterize the optimal policy and present optimal solutions for several special cases. In fact, MAB with switching costs is a special case of the restless bandit problem introduced by Whittle \cite{W88}:
the state of the arm just abandoned changes its state to a ``dummy copy"
state which requires a switching cost if it is to be played.
However, the restless bandit problem is known to be PSPACE-Hard, even to approximate to any non-trivial factor \cite{PT94}.
See the survey \cite{J04}.

We approach the problem from the perspective of approximation algorithms and focus on a finite-time version of MAB, called {\em multi-token Markov game}, introduced 
in an elegant paper by Dumitriu, Tetali and Winkler~\cite{DTW03}. 
In this game, we are also given $n$ Markov chains and each chain has a target which is ultimately reachable.
Each state is associated with a movement cost.
In each time step, the player adaptively chooses one Markov chain to advance with the objective of minimizing the expected total cost for at least one chain to reach its target state.
If there is no switching cost, they show that there is an optimal indexing strategy based on a variant of Gittins index. 
Even if the switching cost is a given nonzero constant,
by a similar argument in \cite{BS94}, one can show that there is no indexing strategy that is optimal (see Appendix~\ref{sec:no_opt_index}).
Hence, in this paper, we study approximation algorithms for 
the multi-token Markov game with switching costs.

\subsection{Problem Definitions and Our Contributions}

We formally define our problem as follows.
We mainly follow the terminology used in ~\cite{DTW03}.
We first introduce the notions for a \MS, which is simply a Markov chain
with (state) movement cost.

\begin{definition}[\textsf{Markov System}~\cite{DTW03}]
A \MS is a tuple $\S=\langle V,P,C,s,t \rangle$, where $V$ is the finite set of states, $P=\{P_{u,v} \}$ is the corresponding transition matrix (a $|V|\times |V|$ matrix), $C=\{C_u\}$ denotes a positive real movement cost for each state $u\in V$, and $s$ (resp. $t$) represents the current (resp. target) state.
We assume that the target is ultimately reachable from every state in $V$,
and we can never exit the target state (so we can set $C_t=0$ and $P_{t,t}=1$). 
If we play $\S$ in state $u\in V$, a cost of $C_u \geq 0$ is incurred and $\S$ transitions from state $u$ to state $v$ with probability $P_{u,v}$.
There is a {\em unit} reward on the target state $t$ that we can collect.
\end{definition}

In the following, we do not distinguish Markov systems and Markov chains,
and use both terms interchangeably.

\subsubsection{Unit Switching Cost}
Now we define our first problem, {\em multi-token Markov game with unit switching cost} (\FindOne).
In this problem, we have a set of (possibly different) \MS and switching from one \MS to another \MS incurs a {\em unit} cost. Our goal is to find a strategy that adaptively chooses the next \MS to play until a unit of reward is collected, and the expected total cost (switching cost plus movement cost) is minimized. Formally, the problem 
is defined as follows.

\begin{definition}[\FindOne]
\label{def:game_findone}
We are given a metric space $\M=(\mathbf{S}\cup \{\mathbf{R}\}, \dist)$ endowed with unit metric (i.e., $\dist(\S,\S')=1$ for any $\S\neq \S'\in \mathbf{S}\cup \{\mathbf{R}\}$).
Each node $\S_i\in \mathbf{S}$ is identified with a \MS 
$\S_i=\langle V_i,P_i,C_i,s_i,t_i \rangle$.
$\mathbf{R}$ is the root node (the initial position) with a unit cost directed edge to every $\S_i$. If we play \MS $\S_i$ in one round and decide to play another \MS $\S_j$ in the next round, we need pay a unit switching cost $\dist(\S_i,\S_j)=1$ in addition to the movement cost.
The game ends when we succeed to make one \MS reach its target state.
\end{definition}

For \FindOne, we provide a simple index strategy
that has a constant approximation ratio.
Here, following the the definition in \cite{BS94} (See also Definition~\ref{def:index} in the appendix), an index strategy means that we can define a suitable index (a real number) for each state
of the Markov chains, and the strategy always chooses to play the Markov chain in which the current state has the minimum index.
%
%, and index strategy means it always chooses the \MS with the smallest index to %advance.}

\begin{restatable}{theorem}{FindoneMS}
\label{thm:Find1MS}
There is a simple index strategy 
that can achieve a constant approximation ratio 
for the \FindOne problem.
\end{restatable}

\noindent
{\bf Our technique:}
In particular, for each state $u$, we create a {\em dummy state} $u'$ 
that connects to $u$ with movement cost $1$.
This captures the unit switching cost.
The index $\index_i$ for $\S_i$ (at state $u$)
is the Gittins index $\gamma_u$ if $\S_i$ is active, and the Gittins index of $u$'s dummy state $u'$ otherwise. 
Our strategy is to simply choose to play the $\S_i$ with the smallest index $\index_i$. 

Although the strategy is extremely simple to state, it is difficult to analyze it directly. Instead, we analyze an alternative strategy, via the doubling framework developed in recent works \cite{ENS18,JLLS20} (Section~\ref{subsec:doubling}). In the doubling framework, we proceeds in phases and in each phase there is an exponentially increasing cost budget. Under this framework, it suffices to show the following guarantee for the budgeted sub-problem: our strategy can succeed
(i.e., collect one unit of reward from some Markov chain) 
with constant probability, under constraint that the total (movement plus switching) cost is below the given budget $B$, which is constant times the optimal cost of the original problem (Lemma~\ref{lem:key_sub_process}). Solving the budgeted sub-problem this is the key technical challenge.

In order to show we can succeed with constant probability under the budget constraint, we consider the set $\Omega_{bad}$ of trajectories in which we fail (do not reach the target).  A trajectory can be naturally partitioned into segments, each being a trajectory in one Markov chain and the switching cost we pay is the number of the segments. Since the cost budget is exhausted without success,
one can show that there is one segment (corresponding to one chain) that the expected cost is large but the success probability is small (conditioning on the former segments). From this, one can argue the grade (Gittins index) of the current chain is much larger than $B$, and the grades of all chains are also at least no smaller than it (due to the greedy rule). By the definition of the grade, it can be roughly understood as the expected movement cost one can hope for reaching the target in this chain. Hence, one can see that 
conditioning on $\Omega_{bad}$, the expected movement cost to reach a target for any strategy is much larger than $B$ (no chain is cheap). But the expected total cost of the optimal strategy is much less than $B$ (recall $B$ is a large constant times OPT). Therefore, one can conclude that the probability of $\Omega_{bad}$ is small.

\subsubsection{General Metric Switching Cost}

Next, we generalize \FindOne to more general metric (where the switching costs form a metric without the restriction to be unit) and more general requirement that we need to collect $K$ units of rewards for any positive integer $K$. We name the new problem \SFindK.

\begin{definition}[\SFindK]
\label{def:SFindK}
We are given a finite metric space $\M=(\mathbf{S}\cup\{\mathbf{R}\},\dist)$ (there is no additional assumption on metric $\dist$). 
Each node $\S_i\in \mathbf{S}$ is identified with a \MS 
$\S_i=\langle V_i,P_i,C_i,s_i,t_i \rangle$. 
Similarly, at the beginning of the game, the player is at the root $\mathbf{R}$, and needs to pay the switching cost $\dist(\mathbf{R},\S_i)$ if he wants to play \MS $\S_i$. 
Switching from $\S_i$ to $\S_j$ incurs a switching cost of  $\dist(\S_i,\S_j)$.
The objective is to adaptively collect at least $K$ units of rewards (making at least $K$ \MS reach their targets), while minimizing the expected total cost (movement cost plus switching cost).
The game ends when we succeed to make $K$ \MS reach their target states.
\end{definition}

\begin{restatable}{theorem}{FindKMS}
\label{thm:GFindkMS_const}
There is a constant factor approximation algorithm  
for the \SFindK problem.
\end{restatable}

\noindent
{\bf Our technique:}
For \SFindK, we also adopt the doubling framework, hence
only need to design an algorithm
$\METASfindK$ (Algorithm~\ref{alg:meta_SfindK}) for the budgeted sub-problem.
$\METASfindK$ should succeed with constant probability using a budget $B$ when $B$ is a constant factor of the optimal cost.
At a high level,
$\METASfindK$ first transform the problem to a \StochKCost instance $\M_{ktsp}$
%We first transform the problem to a \StochKCost instance $\M_{ktsp}$,
by reducing each Markov chain to a related random variable.
Then it applies the {\em non-adaptive} constant factor approximation 
for \StochKCost (developed in \cite{JLLS20}) to obtain an ordering
$\Pi$ of vertices (chains).
We pick a prefix $\Pipref$ such that the switching cost for traversing $\Pipref$ is no larger than a small constant
proportion of the budget. One can show it is possible to collect $K$ units of rewards from $\Pipref$ such that the total movement cost is within the budget with constant probability.

Now, the key is show how to collect $K$ units of rewards from $\Pipref$ such that the movement plus switching cost is within the budget with constant probability. Obviously, ignoring the switching cost, the optimal way (optimal in terms of movement cost) of collecting $K$ units of rewards from $\Pipref$ is to play Gittins index. However, such play may switch back and forth frequently and hence leads to a high switching cost.
To keep the switching cost under control, we insist visiting the chain
in $\Pipref$ one by one and never revisit any chain (hence the switching cost is small). 
However, one may not want to play a chain to the end since the current state is not economical to play and switching to the next chain is a better option. Now the Gittins index comes into rescue. 
We show that there is an interesting threshold $\gamma_{j+1}$ (which is computed from the $K$-th order statistics of suitably defined random variables), such that if the Gittins index of the current state is larger than the threshold, we can give up and decide to switch to the next chain on the $\Pipref$. It turns out such a sequential algorithm (without switching back and forth) can also succeed with constant probability without incurring a much larger movement cost than keeping playing the chain with the smallest Gittins index.

\section{Related Work}
The original paper by  \cite{DTW03} only deals with the $K=1$ case (i.e., one chain reaches its target). 
Recent works \cite{KWW16,GJSS19} observe that it is not difficult to extend their argument to general positive integer $K$ without switching costs.
\cite{KWW16,Sin18} studied the problem 
under richer combinatorial constraints. 
\cite{GJSS19} study a more general problem:
there is a given packing or covering constraint $\calF\subseteq 2^{[n]}$ (e.g., matroid, matching, knapsack) of subsets of chains.
The goal is to make a subset $S$ of chains to reach their targets ($S\in \calF$), while minimizing the dis-utility (with upward-closed constraint) or maximize the utility (with downward-closed constraint).
For semi-additive objective function, they proposed a general reduction which utilizes the ``greedy" algorithms for the problem with full information and they can achieve the same approximation ratio as the "greedy" algorithm does for the full information problem.
%They show if there is an $\alpha$-approx ``greedy'' algorithm when inspection is free, then they can get an $\alpha$-approximation strategy.
%XXX (present their result more precisely. they have result more general than matroid).

Guha and Munagala \cite{GM09} considered two MAX-SNP problems for bandits  with switching costs: {\sl future utilization} and {\sl past utilization}. 
In their problems, each state has a reward and the rewards satisfies the {\em martingale property} (motivated by Bayesian considerations, see their paper for the details). Given two budgets for movement and switching, the future utilization problem aims to make the final reward of the finally chosen chain as large as possible, while the goal of the past utilization is to make the summation of rewards as large as possible. They provided $O(1)$-approximation algorithms for both problems. 
They approached the problem from linear programming with Lagrangian relaxation. Their problems are very different from our problems and it is unclear how to apply their technique to our problems neither.

Our problem is also related to some problems in the stochastic probing 
literature, in particular the classical Pandora's Problem defined in \cite{Wei79}. Suppose there are $n$ closed boxes with independent random rewards (with known distributions). The cost to open box $i$ is $c_i$.
When we open a box, the reward of the box is realized. 
At each time step, we need to decide either to pay some cost to open a new box, or stop and take the box with the maximum rewards.
The goal is to maximize the expected reward minus the opening cost.
Weitzman \cite{Wei79} provided an optimal indexing strategy to this problem.
In fact, one can show the problem is a special case of the Markov Game \cite{DTW03} and Weitzman's index can also be seen as a variant of the Gittin's Index. Recently, the problem has been extent in various ways
(see e.g., \cite{KWW16,Dov18,BK19}). 

%Pandora's Problem only considers to open a single box. What if we need at least five boxes? There can be more complicated combinatorial constraints, and Pandora’s Problem is studied under richer combinatorial constraints in some recent works \cite{KWW16,Sin18}.  Gupta et al. \cite{GJSS19} further considered broader combinatorial constraints and defined a more general model named by {\sl The Markovian Price of Information Model} (Markovian PoI). For semi-additive objective function, they proposed a general method to adapt some ``greedy" algorithms with full information to the PoI world, and the new algorithm can achieve the same approximation ratio as the original "greedy" algorithm.

%Gittins index can be computed in polynomial time.
%, by time-consuming in practice. 
%Recent work \cite{Lat16} provided a computationally efficient approximation of Gittins index with theoretical guarantees. With efficient approximation, it is shown that for multi-armed bandit problem, the index-based algorithm empirically outperforms optimistic algorithms (such as UCB).

Our problem is a stochastic combinatorial optimization problem.
Designing poly-time algorithms for those problems with provable approximation
guarantee has attracted significant attention in recent years 
(see e.g., the survey~\cite{li2016approximation}).
In this paper, we leverage the constant factor approximation algorithm 
for stochastic $k$-TSP \cite{JLLS20} (formally defined later), which is closely related to 
the stochastic knapsack and stochastic orienteering problems. 
In stochastic knapsack, we are given a set of items with random size and profit and a knapsack with fixed capacity. We can adaptively place the items in the knapsack irrevocably, such that the expected total profit is maximized. A variant of the stochastic knapsack has been shown to be PSPACE-hard\cite{DGV04}, and several constant factor approximation algorithms have been developed \cite{DGV04,DGV08,BGK11,GKNR12,LY13}.
Stochastic orienteering \cite{GKNR12} is a generalization of stochastic knapsack, in which there are metric switching costs between different items.
If the total cost is restricted to be no more than $B$, Gupta et al. \cite{GKNR12} provided an $O(\log \log B)$ upper bound of the adaptivity gap, and  Bansal and Nagarajan \cite{BN15} showed a lower bound of $\Omega((\log\log B)^{1/2})$ even when all profits are deterministic.

In online learning literature, there is also a body of work 
\cite{KSU08,Ort08,ADT12,DDKP14,KLM17} studying multi-armed bandit (MAB) with switching costs. However, here playing each arm provides i.i.d. reward, but the underlying distribution is not known. The objective is minimizing the regret.
The challenges and techniques in these settings are completely different.

\section{Preliminaries}
\label{sec:prel}
In this section, we first review the notion of {\em grade} (a variant of Gittins index)
introduced in \cite{DTW03}, then the {\em doubling technique} used in some previous stochastic optimization problems \cite{ENS18,JLLS20}. 
The analysis requires some well known concentration inequalities such as Chernoff Bound and Freedman's Inequality, which are presented in Appendix~\ref{sec:concentration}.
We define some notations which will be used throughout this paper.

\textbf{Notation:} For any (possibly adaptive) strategy $\P$, let $\R(\P)$ be the (random) number of units of rewards $\P$ can collect (i.e. the (random) number of \MS that reach target states). 
Let $\SC(\P)$ and $\MC(\P)$ be the (random) switching cost and movement cost respectively, and let $\TC(\P)=\SC(\P)+\MC(\P)$ be the total cost of the strategy $\P$.
%Let $\ER(\P), \EMC(\P)$, $\ESC(\P)$ and $\ETC(\P)$ be their expectations respectively. 
We say a strategy $\P$ with movement budget (resp. switching budget) $B$ means that $\P$ can pay at most $B$ to advance the \MS (resp. to switch between different systems), 
%say a strategy $\P$ with switching budget $B$ means that $\P$ can pay at most $B$ to switch between different systems 
and say a strategy $\P$ with budget $B$ means that if its movement cost plus switching cost is restricted to be at most $B$.  
Similarly, we say a strategy $\P$ with (movement/switching) budget $B$ in expectation means that the expectation of $(\MC(\P)/\SC(\P))$ $\TC(\P)$ is at most $B$.
%Let \OPT be the optimal strategy.
\subsection{Grade}
\label{sec:gittins_index}
%In this subsection, we review the definition of {\em grade} from \cite{DTW03}.

Give a finite Markov game defined in Dumitriu et al. \cite{DTW03}, we can define the {\em grade} of each state in each \MS. Grade is a slight variant of the original Gittins index defined for the infinite discounted game \cite{Git74}, particularly defined in a very similar way to Weber's prevailing charge \cite{Web92,frostig1999four}. In particular, the grade of state $u$ in \MS $\S$ depends only on $\S$, not other \MS.

\noindent
{\bf A New Game $\S_u(g)$:}
Consider a \MS $\S=\langle V,P,C,s,t \rangle$.
Given a non-negative real number $g\in\mathbb{R}_{\geq0}$ and
the initial state $u$ in $\S$, we define a new game $\S_u(g)$:
For each step, we can quit and end the game, {\bf or} pay the movement cost to advance \S for one step. 
If we reach the target state $t$, we can get $g$ units of profits 
\footnote{We use the term {\em profit} here to distinguish from the term {\em reward} 
(recall we get a unit of reward by reaching a target state).} 
and the game halts. 
The objective is to maximize the objective ``$\mathrm{profit}-\mathrm{cost}$'' in expectation. 
We use $\val(\S_u(g))$ to denote the value of this game, that is 
the expected objective achieved by the optimal strategy,
which we denote by $\OPT(\S_u(g))$. 
Let $R(\OPT(\S_u(g))$ is the (random) indicator that the strategy $\OPT(\S_u(g)$ reaches the target
(e.g., if we quit the game before reaching the target, $R=0$), and 
hence $\E[R(\OPT(\S_u(g)))]$ is the probability that we reach the target state using strategy $\OPT(\S_u(g))$. 
It is easy to see 
$$
\val(\S_u(g))=g \cdot\E[R(\OPT(\S_u(g))]- \E[\MC(\OPT(\S_u(g)))]\geq 0
$$ 
for any $g\in \mathbb{R}_{\geq 0}$, as the strategy can always quit at the very beginning.

\noindent
{\bf Grade $\gamma_u(\S)$:}
In particular, we define the {\em grade} $\gamma_u(\S)$ of state $u$ 
in \MS $\S$ as the unique value of $g$ for which an optimal player is indifferent between the two possible first moves in the game $\S_u(g)$, i.e. he can either play $\S$ for the first step or quit at the very beginning.
We also use $\gamma_u$ and $\S(g)$ as a shorthand for $\gamma_u(\S)$ and $\S_u(g)$ respectively when $\S$ and its current state are clear from the context.

%%%%%%%%%%%%%%%%%%%%%%%%%%%%%% Maybe remove to the appendix
%For any game $\G$, we use $\OPT(\G)$ to represent the optimal strategy for this game $\G$.
To gain a bit more intuition about the grade, 
consider a pure strategy $\ALG$ for the game $\S(g)$.
Note that a pure strategy can be defined by a subset $Q\subset V$ of states: the player chooses to play $\S$, until either target $t$ is reached, or a state in $Q$ is reached and the player chooses to quit, which ends the game immediately.
Let event $\Omega$ be that the token in $\S$ reaches $t$
and let $C_{\S}$ be the (random) movement cost $\ALG$ spends on $\S$.
It is easy to see that 
\begin{gather*}
%\label{eq:def_index}
    \E[g\cdot R(\ALG)-\MC(\ALG)]=g\cdot \Pr[ \Omega]-\E[C_{\S}\mid \Omega]\Pr[\Omega]-\E[C_{\S}\mid \neg \Omega]\Pr[\neg \Omega].
\end{gather*}
One can see it is linear in $g$ for fixed set $Q$.
Hence, the value of the game $\val(\S(g))$
(as a function of $g$) is the maximum of a set of linear functions and is therefore a piece-wise linear convex function in $g$
(See Figure~\ref{fig:gittins_index}).
When $g$ is very small, the optimal strategy should choose to quit immediately, and both the cost and the profit are zero. 
When $g$ is very large, the optimal strategy should never quit before reaching the target.
Hence, as we increase $g$ gradually from $0$, 
there is a point at which we are indifferent between playing $\S$
and quitting immediately, which is the value of the grade for state $u$.
We set $\gamma_t=0$ for the target state $t$.
Readers can refer to \cite{DTW03} for more detailed discussion. 
We can also define a grade for a \MS,
that is the grade of its current state $u$, i.e., $\gamma_u(\S)$.
As shown in \cite{DTW03}, grades can be computed in poly-time (see Section 7 in \cite{DTW03}).

%Let $E_v(\S)$ be the expected movement cost of a non-quitting player starting at $v$ with the subscript sometimes omitted when the player starts at $s$.
%And we have the following Lemma:

%\begin{wrapfigure}{R}%{.5\textwidth}
\begin{figure}[h]
\centering
\includegraphics[width =0.5\textwidth]{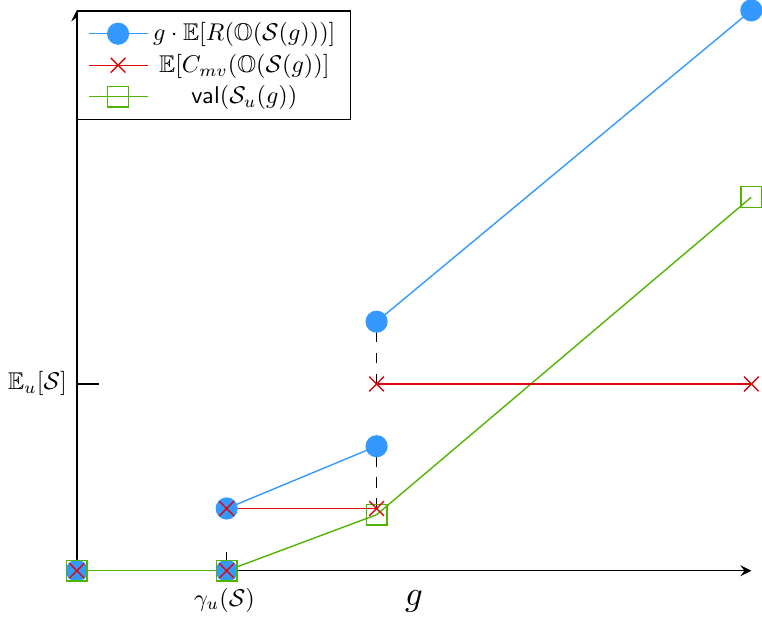}
\caption{An illustration of $\val(\S(g))$, which is a piecewise linear
convex function in $g$. Each piece corresponds to a subset $Q\subset V$.
The first piece corresponds to the empty set $Q=V$ (the player
quits immediately), and the last 
piece corresponds to $Q=\emptyset$ (the player never quits before reaching the target).
$\E_u[\S]$ is the expected cost of a never-quitting player.
In fact, each turning point corresponds to the grade of some state, and the first one 
corresponds to $\gamma_u(\S)$ where $u$ is the current state.
We also show the expected movement cost $\E[\MC(\OPT(\S_u(g)))]$ and the expected profit
$g \cdot\E[R(\OPT(\S_u(g))]$.
%The utility (i.e. ``profit-cost'') does not change at the singularity.
}
\label{fig:gittins_index}
\end{figure}
\vspace{-2mm}
%\end{wrapfigure}

By the definition of the grade, we can get the following lemma easily. 
\begin{lemma}[Optimal solution for $\S(g)$]
\label{lem:optimal_solution_single_GMS}
A strategy for $\S(g)$ is optimal if it chooses to advance \S whenever the grade is no more than $g$ and it chooses to quit whenever the grade is larger than $g$.
When the grade of \S equals to $g$, there is an optimal strategy that chooses to first advance $\S$  and 
there is one that chooses to quit immediately.
%the current state $u$ of \GS satisfies $\gamma_u\leq g$ and it chooses to quit whenever $\gamma_u>g$.
\end{lemma}

We say a game is a {\em fair game} if the value of the game is zero.
An observation that we use repeatedly is that the game $\S_u(\gamma_u(\S))$ is a fair game
(from Figure~\ref{fig:gittins_index}, one can see that $\gamma_u(\S)$ is the largest $g$ such that the game $\S_u(g)$ is a fair game).

Now we define a {\em prevailing cost} \cite{DTW03} and an epoch \cite{GJSS19}. A trajectory is a sequence of states traversed by a player.

\begin{definition}[Prevailing cost]
The prevailing cost of \MS $\S$ in a trajectory $\omega$ is $Y^{\max}(\omega)=\max_{u\in\omega}\gamma_u(\S)$.
\end{definition}
In other word, the prevailing cost is the maximum grade at any point.
In particular, the prevailing cost increases whenever the Markov system reaches a state with grade larger than each of the previously visited states. The prevailing cost can be viewed as a non-decreasing piece-wise constant function of time, which motivates the definition of epoch:
\begin{definition}[Epoch]
An epoch for a trajectory $\omega$ is any maximal continuous segment of $\omega$ where the prevailing cost does not change.
\end{definition}

We also define an interesting teasing game introduced 
in \cite{DTW03}, which is useful later.

\begin{definition}[Teasing game $\S^T$]
\label{def:teasing_game}
Consider the game $\S_s(\gamma_s)$ with initial state $s$.
Whenever the player reaches a state $u$ with grade $\gamma_u>\gamma_s$, we place $\gamma_u$ units of profits at the target state $t$ rather than $\gamma_s$. 
The objective is also to maximize the expectation of 
``profits - costs''.
The $\gamma_u$ profit provides just enough incentive for the player to continue advancing the \S. 
We denote the new teasing game by $\S^T$.
\end{definition}

For the new teasing game, we have the following lemma which also follows directly from the definition of the grade:
\begin{lemma}[Fairness of $\S^T$, Lemma 5.3 in\cite{DTW03}]
\label{lm:fairness}
$\S^T$ is a fair game, and a strategy for $\S^T$ is optimal if and only if the player never quits in the intermediate of an epoch, and only quits at the beginning of an epoch.
%when the current grade is below the current profit at target, and only quits when the current grade equals to the current profit at target.
\end{lemma}

By the above lemma, one can easily see that the expected movement cost of a never-quitting player of the game $\S^T$ is equal to the expected prevailing cost, where being never-quitting means that the player continues playing the system until system reaches the target state and he collects the profits.

\subsection{The Doubling Technique.}
\label{subsec:doubling}
In this subsection, we adopt the {\em doubling technique} which is similar to the ones used in related stochastic optimization problems such as \cite{ENS18} and \cite{JLLS20}.
See the pseudo-code of the framework in Algorithm~\ref{alg:alg_general}.
Basically, the framework proceeds in phases, and in $i$th phase, we call a sub-procedure denoted by $\ALGmeta(\M_{i-1},k_{i-1},B_i)$
in which we start with $\M_{i-1}$, the current state of all \MS, and aim at collecting the remaining $k_{i-1}$ units of reward with total cost budget $B_i=O(1)\beta^i$ (in expectation).

\begin{algorithm2e}[H]
\caption{A general algorithm  $\ALGgeneral$}
\label{alg:alg_general}
{\bf  Input:} The problem instance $\M$, objective number of rewards $K$\\
{\bf Process:}\\
 Set $\beta \in (1,2)$, $k_0=K$, $\M_0=\M$, $c=O(1)$\;
 \For{phase~$i=1,2,\cdots$}
 {
  $(\M_{i},k_{i})\leftarrow \ALGmeta(\M_{i-1},k_{i-1},B_i=c\beta^i)$\;
   \If{$k_i\leq 0$}{\textbf{Break}}
 }
\end{algorithm2e}
%For each phase $i$, $\ALGmeta(\M_{i-1},k_{i-1},B_{i})$ begins to play the game at root $\mathbf{R}$ and can stop anywhere without breaking the cost budget. This requirement makes $\ALGgeneral$ cost at most $2B_i$ for each phase, $B_i$ for running $\ALGmeta(\M_{i-1},k_{i-1},B_i)$ and the other $B_i$ for switching back to root $\mathbf{R}$ for the next phase.}

Recall that we have a unique root $\mathbf{R}$.
In particular, in \FindOne, we let $\ALGmeta(\M_{i},k_{i},B_{i+1})$ begin to play the game at the \MS where $\ALGmeta(\M_{i-1},k_{i-1},B_{i})$ halts, that's $\ALGmeta(\M_{i},k_{i},B_{i+1})$ does not need pay the unit switching cost for the \MS where $\ALGmeta(\M_{i-1},k_{i-1},B_{i})$ halts.
In \SFindK, we require that the strategy goes back to $\mathbf{R}$ after $\ALGmeta(\M_{i-1},k_{i-1},B_{i})$ halts for simplicity of the analysis, and hence
the next phase $\ALGmeta(\M_{i},k_{i},B_{i+1})$ starts the game at the root $\mathbf{R}$.
%Then $\ALGgeneral$ costs at most $2B_{i+1}$ for $i$-th phase: $B_{i+1}$ for running $\ALGmeta(\M_{i},k_{i},B_{i+1})$ and the other $B_{i+1}$ for 
This blows up the total cost by at most a factor of 2 since
switching back to the root $\mathbf{R}$ costs at most $B_i$.
The main reason of doing so is 
to avoid the case where 
$\ALGmeta(\M_{i-1},k_{i-1},B_{i})$ stops at a \MS far-way from the other chains and $\ALGmeta(\M_{i},k_{i},B_{i+1})$ has to pay a lot in the first switch (rather this switching cost is amortized to the $i-1$th phase).

%The difference between two problems is that in $\SFindK$, we want to avoid the case when $\ALGmeta(\M_{i-1},k_{i-1},B_{i})$ stops at a \MS far-way from the others and it is unfair for $\ALGmeta(\M_{i},k_{i},B_{i+1})$ to waste budget on going back to the root.

To analyze the algorithm framework, we define $\OPT(\M,k)$ to 
be the optimal strategy for the problem instance $\M$ (the game starts from the root $\mathbf{R}$)  with the target number of rewards $k$.
%To simplify the analysis, we require that both strategies $\OPT(\M_i,k_i)$ and $\ALGmeta(\M_{i},k_{i},B_{i+1})$ start the game at the root $\mathbf{R}\in\M_{i}$ at the beginning of each phase.
%This requirement makes $\ALGgeneral$ cost another $B_{i-1}$ switching cost at most for the $i$-th phase.
%As $\ALGgeneral$ goes back to the root $\mathbf{R}$ at the beginning of phase $i$, there is no need to know which is the last \MS played by $\ALGgeneral$  in the $(i-1)$th phase.
Intuitively, for $j\geq i$, the expected cost of $\OPT(\M_j,k_j)$ is no more than the one of $\OPT(\M_i,k_i)$ as $k_j\leq k_i$, i.e. in some sense $\M_j$ can get $k_i-k_j$ units of rewards for free. 
%We note that $\M_i$ can be determined by $\omega_{i}$.
We can prove this formally and get the following lemma (whose proof can be found in Appendix~\ref{sec:larger_better}).

\begin{lemma}
\label{lm:larger_better}
For any $j\geq i\geq 1$ and any Algorithm $\ALGmeta$, one has 
\begin{align*}
    \E[\TC(\OPT(\M_i,k_i))]\geq \E[\TC(\OPT(\M_j,k_j))].
\end{align*}
Notice that the randomness is over the entire run of 
$\ALGgeneral$. 
%the randomness of which is from the first $i$ phases and the randomness of  is over the phases after $i$.
\end{lemma}

We use Lemma~\ref{lm:larger_better} to prove the following Lemma~\ref{lem:key_sub_process}, which is used for both \FindOne and 
\SFindK.

\begin{lemma}[Budgeted Subproblem]
\label{lem:key_sub_process}
We are given an \SFindK (or \FindOne) instance $\M$ with positive integers $k$ and $B\in \mathbb{R}_{\geq 0}$. 
For any $B> c_1\E[\TC(\OPT(\M,k))]$, if there is an algorithm $\ALGmeta$ that can succeed in collecting $k$ units of rewards with some constant probability (say more than $0.01$)
using expected total cost at most $c_2B$,
where $c_1,c_2$ are some universal constants, then there is a constant factor approximation algorithm for \SFindK (or \FindOne).
\end{lemma}

The proof is similar to the previous ones \cite{ENS18,JLLS20} with some subtle modifications and can be found in Appendix~\ref{sec:appendix_doubling}.
%Roughly, it says that $\ALGgeneral$ is a constant approximation algorithm if it can ensure that whenever $u_i^*(\omega_{i-1})$ is small (i.e. $\OPT$ has a large success probability within budget $\beta^i$) and $10\ETC(\OPT(\M_i,k_i))(\omega_{i-1})< \beta^i$ (i.e. budget $\beta^i$ is large enough for $\OPT(\M_i,k_i)$), $\ALGgeneral$ can also succeed with a constant probability in the first $i$ phases (i.e. with budget $O(1) \cdot \beta^i$).Then we can get the following lemma directly and reduce all the problems into designing the key sub-process \ALGmeta.

If we can design $\ALGmeta$ which satisfies the precondition of Lemma~\ref{lem:key_sub_process}, then by running $\ALGmeta$ in the framework, we get an $O(1)$-approximation algorithm. 
All our effort goes into designing $\ALGmeta$ that satisfies the precondition of Lemma~\ref{lem:key_sub_process} in the following sections.

\section{Markov Game with Unit Metric}
\label{sec:unit}
In this section, we consider the \FindOne problem (Definition \ref{def:game_findone}). 
Our main result is Theorem~\ref{thm:Find1MS}.
\FindoneMS*

\begin{algorithm2e}[H]
\caption{Algorithm for \FindOne}
\label{alg:alg_true_findone}
{\bf  Input:} The instance $\M$\\
%{\bf Process:}\\
\While{we have not collected any reward}
{
Choose to play \MS $\S_i$ with $i=\arg\min_{i}\index_i$; 
}
\end{algorithm2e}

Since our goal is to design an indexing strategy, we
need to define an index which can incorporate the information of  switching cost.
In particular, for each \MS $\S_i \in \mathbf{S}$ and every state $u$ of  $\S_i$, we create a {\em dummy state} $u'$ for $u$ with unit movement cost ($C_{u'}=1$) and deterministic transition to $u$ ($P_{u',u}=1$). 
The dummy states are used to capture the unit switching cost.
Let $\gamma'_u$ denote the grade of $u$'s dummy state, and call $\gamma'_u$ the {\em dummy grade} of $u$. 
We say a \MS $\S_i$ is  {\em active} at a time step $t$ if $\S_i$ is played in the previous step 
(i.e., continuing to play $\S_i$ in the $t$-th step does not incur any switching cost), and {\em inactive} otherwise.

Initially, we are at the root $\mathbf{R}$, and all of the \MS are inactive.
Now we define a grade $\index_i$ for $\S_i$ which is at state $u$: 
if $\S_i$ is active, then its $\index_i$ is defined to be the grade $\gamma_u$ of $u$; otherwise,  $\index_i$ is defined to be the dummy grade of $u$, which we denote by $\gamma'_u$. 

Our strategy for \FindOne is simply choosing to play the $\S_i$ with the smallest grade $\index_i$, breaking ties arbitrarily. See also Algorithm~\ref{alg:alg_true_findone}. 

Although Algorithm~\ref{alg:alg_true_findone} is simple to state,
directly analyzing it seems difficult.
Rather, we analyze an alternative algorithm via the doubling technique framework. 
More specifically, we apply the framework $\ALGgeneral$ (Algorithm~\ref{alg:alg_general}), which proceeds in phases. 
In each phase, it calls the sub-procedure $\METAfindone$ (Algorithm~\ref{alg:meta_findone}).
In the sub-procedure, we have a movement cost budget and a switching cost budget. We repeatedly play the \MS with the smallest grade, until one of the budgets is exhausted or all \MS reach their target states. 
Note that we may not stop immediately when we reach a target state and collect one unit of reward.
Instead, we should remove the present system and keep on playing (if there are still budget and available systems in this phase). Hence, the cost of the alternative algorithm is no less than that of Algorithm~\ref{alg:alg_true_findone}.
See also Appendix~\ref{sec:alg_findone}.

%The pseudo-code can be found in Algorithm~\ref{alg:alg_analysis_findone} and \ref{alg:meta_findone}. 

%\haotian{Please keep only the algorithm you are going to formally analyze in this section. Put the simple algorithm in the intro for motivation. It looks very weird to describe an algorithm and say you are not going to analyze it.}

%\haotian{I just realized that Algorithm 3 and Algorithm 4 are essentially the same algorithm, only stated differently. We need to make this clear.}

\subsection{Analysis.}
Recall that 
%Algorithm~\ref{alg:alg_analysis_findone} is simply an instantiation of the sub-procedure of Algorithm Framework~\ref{alg:alg_general} in Section~\ref{sec:prel}.
in order to prove the main result of this section (i.e. Theorem~\ref{thm:Find1MS}), 
it suffices to show that the sub-procedure $\METAfindone$ (Algorithm~\ref{alg:meta_findone}) satisfies the precondition of Lemma~\ref{lem:key_sub_process}. 
The key is to prove the following lemma, which is the precondition of Lemma~\ref{lem:key_sub_process} 
specialized for \FindOne.
Recall that we are still aiming at solving \FindOne where Algorithm~\ref{alg:alg_true_findone} stops immediately whenever it makes one \MS reach its target state. The alternative Algorithm~\ref{alg:meta_findone} may collect more than one unit of reward and it is only used for analysis and setting an upper bound for the expected cost of our true algorithm (Algorithm~\ref{alg:alg_true_findone}).

\begin{algorithm2e}[t]
\caption{Subprocedure $\METAfindone$}
\label{alg:meta_findone}
{\bf  Input:} The instance $\M$, budget $2^{8}B$\\
%{\bf Process:}\\
Set movement budget $2^{7}B$ and switching budget $2^{7}B$\;
Set $k\leftarrow 1$\;

\While{there are available \MS and condition ${\cal A}$ holds}
{
    Choose to play \MS $\S_i$ with $i=\arg\min_{i}\index_i$\; 
    \If{We reach a target state $t$ in $\S_i$}
    {
        Collect one unit reward and mark $\S_i$ as {\em unavailable};\\
        $k\leftarrow k-1$\;
        %{\color{red} Haotian: Isn't $R_t=1$ in the unit reward case?}
    }
}
{\bf Return:} The updated instance $\M$, the remaining number of target states $k$\;
\BlankLine
{\bf Define:} Condition ${\cal A}$:\label{ln:condition_A}\\
The next move does not make the total movement cost or switching cost exceed $2^{7}B$\;
\end{algorithm2e}

\begin{lemma}
\label{lm:precondition}
For any input $\M$, let $\OPT$ be the optimal strategy for this instance.
If $B\geq 10\E[\TC(\OPT)]$, with probability at least 1/20,
\METAfindone (Algorithm~\ref{alg:meta_findone}) can collect at least one unit reward with budget $2^{8}\cdot B$.  
\end{lemma}

%The index-based algorithm $\METAfindone(\M,B)$ is pretty simple: For each step, choose an arbitrary \MS with the smallest index to play until run out the budget or there is no available \MS. We allocate $2^{7}B$ and $2^{100 }B$ budgets to $\METAfindone$'s switch and play, respectively. It is notable that the main algorithm $\ALGfindone$ for \FindOne is: For each step, choose an arbitrary \MS with the smallest index to play.

\begin{proofof}{Lemma~\ref{lm:precondition}}
%It suffices to prove that \METAfindone satisfies the precondition of Lemma~\ref{lem:key_sub_process}. 
Due to Condition~${\cal A}$ (Line~\ref{ln:condition_A}),
the total cost budget cannot be violated.
Hence, it suffices to prove the success probability 
is at least 1/20.
\begin{comment}
As $B\geq 10\E[\TC(\OPT)]$, by Markov inequality, we know that
\[
    \Pr[\TC(\OPT)\leq B]\geq 0.9.
\]
\end{comment}

%We define some random variables below to analyze $\METAfindone$. 
%The moments when $\METAfindone$ switches, i.e. when \METAfindone %decides to pay the unit cost and to advance another inactive \MS, %play key roles in the analysis.

We only need to consider the case when there are more than one system, otherwise it is optimal to switch the only system and make it reach the target state.
Suppose $\METAfindone$ decides to play some \MS $\S_{i_{j}}$ after switching $j$ times, i.e., the sequence of chains played by
$\METAfindone$ 
is $(\mathbf{R},\S_{i_1},\S_{i_2},\cdots,\S_{i_j},\cdots)$.
Note that $\METAfindone$ may revisit a chain (i.e., $\S_{i_t}=\S_{i_j}$ for some $i_t\ne i_j$).
Let $\omega_{j}$ be the path (a sequence of states) \METAfindone traverses on $\S_{i_j}$ after it switches to $\S_{i_j}$ and before it switches to $\S_{i_{j+1}}$ (or stops due to running out budget).
We define the stopping time $\tau$ as the numbers of switching of \METAfindone when it halts (equivalently, the switching cost of \METAfindone).
For $\tau+1\leq j\leq 2^{7}B$, we let $\omega_j=\emptyset$. 
Let $\omega=\cup_{j\geq 0}^{2^{7}B}\omega_j=\cup_{j\geq 0}^{\tau}\omega_j$ be the whole trajectory traversed by $\METAfindone$ and let 
$\Omega$ denote the set of all possible trajectories 
$\METAfindone$ can traverse.
We use the notation $\omega_{[0:j]}=\cup_{t=0}^{j}\omega_t$ to denote the prefix of $\omega$.

Now we define the Boolean random variables $X_{j}$
for $j\leq \tau$:
%what has happened until the $j$-th switching occurs, i.e. the sequence of  that $\METAfindone$ observes at $\S_{i_1},\cdots,\S_{i_{j-1}}$, as follows: 
%After switching $j$ times, 
if \METAfindone can get the reward in $\S_{i_j}$, then $X_{j}=1$. Otherwise (i.e., \METAfindone switches out or runs out the budget) $X_{j}=0$. 
If $j>\tau$, we let random variables $X_{j}=0$.
Similarly, for $j\leq \tau$, 
we define random variable $C_{j}$ to represent the movement cost \METAfindone spends on $\S_{i_j}$ 
($C_j$ does not include the unit switching cost).

%between $j$-th switch and $(j+1)$-th switch, or sample path on $\S^*$ after $j$-th switch if $\METAfindone$ ends on $\S^*$ because of running out the budget.

 Let $\index_{i_j}$ be the smallest grade of all systems when the $j$-th switch occurs, which is the dummy grade of the current state of $\S_{i_j}$.
 We use $s_{i_j}$ to denote the current state of $\S_{i_j}$ when the $j$-th switch
 occurs, and $s'_{i_j}$ be its dummy state.
 Hence, $\index_{i_j}=\gamma_{s'_{i_j}}(\S_{i_j})$.
 By the greedy process of \METAfindone, we know that 
 \begin{align*}
     \index_{i_1}\leq \index_{i_2}\leq \cdots\leq \index_{i_j}\leq \cdots\leq \index_{i_\tau}.
 \end{align*}
 
 %And we use $C'_{j}=C_{j}+1$ to capture the unit switching cost, and equivalently, the ``movement cost'' algorithm spends when it starts from the dummy state.
 %plus the unit switching cost for the $j$-th switch, where we can treat the unit switching cost as the movement cost from the dummy state $u'$ to $u$. 
%Standard, define filtration $\F_n=\sigma(X_0,\cdots,X_n)$ from the process.
For a trajectory $\omega$,
let $\T_j(\omega)=\E[X_j\mid \omega_{[0:j-1]}]$.
Indeed, one can see that $\T_j$ is a random variable with randomness
from $\omega_{[0:j-1]}$.
Let $\T=\sum_{j=1}^{\tau}\T_j$ and $\T(\omega)=\sum_{j=1}^{\tau}\T_j(\omega)=\sum_{j=1}^{\tau}
\E[X_j\mid\omega_{[0:j-1]}]$. 
 %Let $T_j(\omega)=\E[X_j\mid \F_{j-1}](\omega_{[0:j-1]})$, and $T(\omega)=\sum_{j=1}^{2^{7}B}T_j(\omega)=\sum_{j=1}^{2^{7}B}(\E[X_j\mid \F_{j-1}]( \omega_{[0:j-1]})) =\sum_{j=1}^{2^{7}B}\E[X_j\mid \omega_{[0:j-1]}]$.
 \begin{claim}
 \label{clm:bound_grade}
 We know that
 \begin{align*}
     \Gamma_{i_{j+1}}\geq \frac{\E[C_j+1\mid \omega_{[0;j-1]}]}{\T_j(\omega)}
     =\frac{\E[C_j+1\mid \omega_{[0;j-1]}]}{\E[X_j\mid \omega_{[0:j-1]}] }\geq \Gamma_{i_j}.
 \end{align*}
 \end{claim}
 
 \begin{proof}
 Recall that the algorithm plays $\S_{i_j}$ whose grade is $\Gamma_{i_j}$ when the $j$-th switch occurs. 
 Note that we know the value of $\Gamma_{i_{j+1}}$ when the $j$-th switch occurs.
 Denote $U_j=\{u\in V_{i_j}\mid \gamma_u\leq \Gamma_{i_{j+1}}\}$.
 \METAfindone continues playing the system $\S_{i_j}$ until it reaches the target state or reach a state outside $U_j$ when the cost budget is not exhausted.
 
We know that $\S_{s_{i_j}'}(\index_{i_j})$ is a fair game where $s_{i_j}'$ is the dummy state of $s_{i_j}$.
On one hand, by Lemma~\ref{lem:optimal_solution_single_GMS}, the strategy determined by $U_j$ may not be the optimal solution to the game $\S_{s_{i_j}'}(\index_{i_{j}})$, so we have
\begin{align*}
    \E[X_j\mid \omega_{[0:j-1]}]\cdot\index_{i_j}-\E[C_j+1\mid \omega_{[0:j-1]}]\leq \val(\S_{s_{i_j}'}(\index_{i_{j}}))=0.
\end{align*}

On the other hand, as $\index_{i_{j+1}}\geq \index_{i_j}$, 
the game $\S_{s_{i_j}'}(\index_{i_{j+1}})$ is better than fair, hence
$\val(\S_{s_{i_j}'}(\index_{i_{j+1}}))\geq 0$. 
The strategy determined by $U_j$ is exactly the optimal solution to the game $\S_{s_{i_j}'}(\index_{i_{j+1}})$ according to Lemma~\ref{lem:optimal_solution_single_GMS}.
Thus we can conclude 
\begin{align*}
    \E[X_j\mid \omega_{[0:j-1]}] \cdot \index_{i_{j+1}}
    -\E[C_j+1\mid \omega_{[0:j-1]}]=\val(\S_{s_{i_j}'}(\index_{i_{j+1}}))\geq 0.
    %\index_{i_{j+1}}\geq \frac{\E[C_j+1\mid \omega_{[0:j-1]}]}{\T_j(\omega)}.
\end{align*}
This finishes the proof of the claim.
 \end{proof}
 
Recall that it suffices to show the success probability of \METAfindone is at least 1/20 when $B\geq 10\E[\TC(\OPT)]$.
%To prove Lemma~\ref{lm:precondition}, we need two other results.
The following key lemma states that with probability at least 3/20 that either \METAfindone succeeds to collect at least one unit of reward, or the summation of the conditional expectation is large. 
%conditioning on which can imply \METAfindone can succeed with high probability.
\begin{lemma}
\label{lm:large_ET} (Key Lemma)
%Conditioning on that \OPT can succeed with total cost no more than $B$ and probability no less than 0.9, then 
With probability at least 3/20, either $\sum_{j=1}^{\tau}\T_j\geq 2$ or $\sum_{j=1}^{\tau}X_j\geq 1$.
Equivalently, one has
\begin{align*}
    \Pr\left[\sum_{j=1}^{\tau}\T_j\geq 2\vee \sum_{j=1}^\tau X_j\geq 1\right]\geq 3/20.
\end{align*}
\end{lemma}

\begin{proof}
Recall that $\T_j=\E[X_j\mid{\cal F}_{j-1}]$, $\T=\sum_{j=1}^{\tau}\T_j$,
and $\Omega$ is the set of all possible trajectories.
To prove the statement, we bound the probability of a few bad cases: 

\noindent \textbf{Case (1):} [$\T< 2$, $\sum_{j=1}^\tau X_j=0$ and the switching budget runs out first.]

Let $\Omega_1=\{\omega\in \Omega \mid \T(\omega)=\sum_{j=1}^{\tau}\T_j(\omega)<2,\sum_{j=1}^\tau X_j (\omega)=0, \tau=2^{7}\cdot B \}$ be the set of trajectories corresponding to Case (1). 

%\[  \Pr[\OPT \text{ succeeds within Budget $B$}\mid  \Omega_1 ] \]

Conditioning on any sample path (trajectory) $\omega\in \Omega_1$, there must exist $j\in[2^{7}\cdot B]$ such that $\T_j(\omega)\leq \frac{2}{2^{7}B}= 2^{-6}\cdot B^{-1}$.
%, i.e. $\E[X_j\mid \omega_{[0:j-1]}]\leq 2^{-6}\cdot B^{-1}$. 
%We must pay the unit switching cost, which means that $\E[C_j+1\mid \omega_{[0:j-1]}]\geq 1$.
By Claim~\ref{clm:bound_grade}, we know that 
$$
\Gamma_{i_{j+1}}\geq \frac{\E[C_j+1\mid \omega_{[0:j-1]}]}{\T_j(\omega)}
\geq \frac{1}{\T_j(\omega)}
\geq  2^6\cdot B.
$$

\iffalse
Define a new \MS $\S_{i_j}|U_j=<U_j\cup\{s\}\cup T,P',C,s,t>$ with 
\begin{align*}
    P'_{u,s}=P_{u,s}+\sum_{v\in V\setminus\left\{U_j \cup\left\{s,t\right\}\right\}}P_{u,v} && \text{for } u\in U_j
\end{align*}
and $P'_{u,w}=P_{u,w}$ for $w,v\in U_j$. 

In fact, $\S_{i_j}|U_j$ means that if you leave $U_j\cup\left\{s,t\right\}$ before reaching target state $t$, you need to go back to the starting state $s$.
\fi

%As \METAfindone chooses the \MS with the smallest grade to play, then the grades of all \MS do not decrease at the moment of switching, and the grades of all (inactive) \MS after $j$-th switching are at least $\E[C_j+1\mid  \omega_{[0:j-1]}]/\E[X_j\mid \omega_{[0:j-1]}]\geq 2^{9}\cdot B$ by Lemma~\ref{lem:giitins_one_round_generally_P2}.

%Condition on any sample path $\omega$ in this case, and 

Recall that $\M$ is the initial instance accepted by \METAfindone, and suppose the $\M'$ is the resulting instance at the end of \METAfindone (their difference can be determined by $\omega$). 

Conditioning on $\omega$,
we let $\OPT'$ be the optimal strategy for instance $\M'$.
If \METAfindone does not collect any reward, we still need to collect one unit reward from $\M'$. If \METAfindone succeeds to collect at least one unit of reward, $\OPT'$ does nothing and halts immediately. 
Note that $\E[\TC(\OPT')\mid \omega]$ is a random variable with randomness from $\omega$.

%Let $k'$ be the remaining number of reward we need to collect from $\M'$ ($k'\leq 1$) and let $\OPT'$ be the optimal strategy for instance $(\M',k')$. 
%\textcolor{red}{Notice that for if  $``k'\leq 0"$ happens conditioning on a sample path $\omega$, it means \METAfindone succeeds to make at least one \MS reach their target states on $\omega$, and hence $\OPT'$ does not need to collect any reward and $\TC(\OPT')=0$ for the corresponding $\omega$.}
By Lemma~\ref{lm:larger_better}, we have 
\[
\E[\TC(\OPT')]\leq \E[\TC(\OPT)]\leq B/10.
\]

Recall that the smallest grade of all available \MS does not decrease and thus grades of all \GS are at least $2^{6}\cdot B$ when \METAfindone halts.
We want to bound the probability $\Pr[\TC(\OPT')\leq B]$.
We need the following claim:
\begin{claim}
\label{clm:large_index_low_reward}
For any positive real number $B$ and any \SFindK instance $\M$, if the grades of all \MS are at least $\zeta\cdot B$ for some constant $\zeta\geq 1$, then for any strategy $\ALG$ with total cost budget $B$, one has
\begin{align*}
    \Pr[\ALG \text{ succeeds to collect at least one unit reward } ]\leq 1/\zeta.
\end{align*}
\end{claim}
\begin{proof}
%We are considering $\M$ under the rule of the problem \SFindK, and suppose there are $m$ \MS $\S_1,\cdots,\S_m$. 
%Recall the definition of the game $\S_u(g)$ when we define the grade (See also Section~\ref{sec:gittins_index}).
With the input $\M$, let $\Gamma(\S_i)$ be the grade of the inactive $\S_i$ and $u_i'$ be the dummy state of the initial state $u_i$ of $\S_i$.
We first simplify the game by only requiring to pay the unit switching cost once
for each \MS.
Equivalently, only the initial state $u_j$ has its corresponding dummy state $u_j'$ in $\S_j$, and when the player decides to switch to $\S_j$ again, she switches to 
the current normal state (instead of the dummy state).
\footnote{Hence, the simplified game can be reduced to a Markov game
without switching cost, hence solved optimally.}
Consider an arbitrary strategy $\ALG$.
Obviously, the cost of $\ALG$ is only less in this simplified game
(trajectory-wise). 

Now, imagine that we run $\ALG$ on a composition game 
$\CGame=\S_{u_1'}(\Gamma(\S_1))\circ \S_{u_2'}(\Gamma(\S_2))\circ\cdots\circ \S_{u_m'}(\Gamma(\S_m))$: The new composition game has the same set of Markov
chains; hence a trajectory in the original game is also a trajectory
in the new game.
We know $\S_{u_i'}(\Gamma(\S_i))$ is a fair game and any combination (simultaneous, sequential, or interleaved) of independent fair games is still fair (e.g., Lemma 5.4 in \cite{DTW03}).
Hence, $\CGame$ is a fair game.

We use $p_j$ to denote the probability that $\ALG$ makes system $\S_j$ reach its target state and $C_j$ be the cost that $\ALG$ spends on $\S_j$ in $\CGame$.
Indeed, one can see $C_j$ is also the cost $\ALG$ spends on $\S_j$ in the simplified game.
%
%Thus, $C_j$ is the movement cost that $\ALG$ spends on $\S_j$ plus an indicator that $\ALG$ goes to $\S_j$ at least one time.
%By Lemma~\ref{lem:optimal_solution_single_GMS}, we have that $p_j\Gamma(\S_j)-\E[C_j(\ALG)+1]\leq 0$.
As the composition game is fair, thus we know the expected ``profits-cost'' of 
$\ALG$ is at most 0.
In particular, one has
$\sum_{j=1}^m \Gamma(\S_j)p_j- \sum_{j=1}^m\E[C_j]\leq \val(\CGame) = 0$.

%for each step in \SFindK, whichever system is played by $\ALG$, then let $\ALG'$ play the corresponding chain for one step in the composition game $\S_{u_1'}(\Gamma(\S_1))\circ \S_{u_2'}(\Gamma(\S_2))\circ\cdots\circ \S_{u_m'}(\Gamma(\S_m))$. 
%If $\ALG$ runs out of its budget, then $\ALG'$ chooses to quit and end the composition game.

%To make the systems independent of each other, we relax the setting and let each system $\S_i$ only charges switching cost of $\ALG$ once, i.e. for any system $\S_i$, $\ALG$ can pay unit switching cost and make it active forever.
%With this assumption, any system does not need to go back to its dummy state after the player switches out.

For each $j\in[m]$, one has $ \Gamma(\S_j)\geq \zeta\cdot B$. 
%and $\Pr[\ALG \text{ succeeds to collect at least one unit reward} ]\leq\sum_{j=1}^{m}p_j$. 
As the budget is $B$, $\sum_{j=1}^m\E[C_j]\leq B$. Hence, we have that
\begin{align*}
    \Pr[\ALG \text{ succeeds to collect at least one unit reward } ]
    \leq  \sum_{j=1}^m p_j
    \leq  \frac{B}{\min_j\Gamma(\S_j)}
    \leq  \frac{B}{\zeta \cdot B}
     =  1/\zeta.
\end{align*}
This finishes the proof of the claim.
\end{proof}
%Suppose $\OPT'$ decides to play some system $\S$ whose index is $\Gamma\geq 2^{6}\cdot B$.We know that $\S(\Gamma)$ is a fair game and thus the expected number of rewards $\OPT'$ can get within budget $B$ is at most $\frac{B}{\Gamma}\leq 2^{-6}$.
Note that the randomness of $\OPT'$ comes from $\omega$.
Applying Claim~\ref{clm:large_index_low_reward} to $\OPT'$ with total cost budget $B$, we can see 
\[
\Pr[\OPT' \text{ succeeds to collect at least one unit reward} \wedge \TC(\OPT')\leq B \mid \omega\in \Omega_{1}]\leq 2^{-6},
\]
which means that $ \E[\TC(\OPT')\mid \omega \in \Omega_1]\geq B(1-2^{-6})\geq 0.9B$. Combining the fact that $\E[\TC(\OPT')]\leq B/10$, we know that 
\[
\Pr[\omega\in \Omega_1]\leq 1/9.
\]

\noindent \textbf{Case (2):} $\T<2$, $\sum_{j=1}^\tau X_j=0$ and the movement budget runs out first. 
We define some additional variables to analyze this case.
Let $\xi(\omega)$ be the movement cost of the (next) move which breaks Condition ${\cal A}$ and makes \METAfindone halt.
%Recall we use $\E[C_j\mid \F_{j-1}]$ to be the conditional expected movement cost of $\METAfindone$ spends on $\S_{i_j}$, and $U_j$ is a subset of states after reaching which $\METAfindone$ can continue playing. 
%Here we use $\E[C_j(U_j)\mid \F_{j-1}]$ to be the conditional expected movement cost the algorithm spends if player continues playing $\S_{i_j}$ until it reaches the target state or reach a state outside $U_j$ without the constraint of condition ${\cal A}$.
We divide Case (2) further into the following cases.

{\bf Sub-case (2.1):}[Large breaking cost.]
This Sub-case corresponds to the set $\Omega_{2.1}=\{\omega\in \Omega: \T(\omega)<2,\sum_{j=1}^{\tau}X_j(\omega)=0, \sum_{j=1}^{\tau} C_j(\omega)+\xi(\omega) \geq 2^{7}\cdot B, \xi(\omega)\geq 2^6\cdot B\}$.
As $\xi(\omega)\geq 2^{6}\cdot B$, conditioning on $\omega\in \Omega_{2.1}$, it is impossible for $\OPT'$ to make $\S_{\tau}$ reach the target state within budget $B$.
Here we say $\OPT'$ does something within budget $B$, means the event that $\OPT'$ 
does something and its total cost is no more than $B$ at the time of accomplishment.

For those $\omega\in \Omega_{2.1}$, as $\xi(\omega)\geq 2^{6}\cdot B$, we know \METAfindone must have decided to play some \MS whose grade is at least $2^{6}\cdot B$.
More specifically, let $u$ be final current state of $\S_{\tau}$ when \METAfindone halts, and by the definition we have $C_u=\xi(\omega)$.
By the definition, we know that $\gamma_u\geq C_u\geq 2^6\cdot B$.
We define $\S_{i_{\tau+1}}$ to be the (inactive) system with the second smallest grade when the $\tau$-th switch occurs ($\S_{i_{\tau}}$ has the smallest grade then), and $\Gamma_{i_{\tau+1}}$ is the corresponding grade.
One can see that $\Gamma_{i_{\tau+1}}\geq \gamma_u\geq 2^6\cdot B$. 

Similar to {\bf Case (1)}, by Claim~\ref{clm:large_index_low_reward}, one has
\[
\Pr[\OPT' \text{ succeeds to collect at least one unit reward within Budget $B$}\mid \omega\in \Omega_{2.1}]\leq 2^{-6},
\]
which implies that
$\Pr[\omega\in \Omega_{2.1}]\leq 1/9$ by the fact that $\E[\TC(\OPT')]\leq B/10.$

%As \METAfindone runs out the playing budget $2^{7}\cdot B$, i.e. $\sum_{j=1}^{\tau} C_j(\omega)\geq 2^{100}\cdot B$, the probability of this Sub-case is at most $2^{-90}$, i.e. $\Pr[\text{Sub-case(2.1)}]\leq 2^{-90}$.
%Suppose the expected total movement cost of \METAfindone is no more than $2^{7}\cdot B$, i.e. $\sum_{j=1}^{\tau}\E[C_j\mid {\cal F}_{j-1}]\leq 2^{7}\cdot B$. As \METAfindone runs out the playing budget $2^{100}\cdot B$, i.e. $\sum_{j=1}^{\tau} C_j\geq 2^{100}\cdot B$, the probability of this Sub-case is at most $2^{-90}$, i.e. $\Pr[\text{Sub-case(2.1)}]\leq 2^{-90}$.

{\bf Sub-case (2.2):}[Small expected movement cost, small breaking cost.]
Particularly, this Sub-case corresponds to the set $\Omega_{2.1}=\{\omega\in \Omega: \T(\omega)<2,\sum_{j=1}^{\tau}X_j(\omega)=0, \sum_{j=1}^{\tau}\E[C_j\mid \omega_{[0:j-1]}] \leq 2^{4}\cdot B, \sum_{j=1}^{\tau} C_j(\omega)+\xi(\omega) > 2^{7}\cdot B, \xi(\omega) \leq 2^6\cdot B\}$.
Obviously, if $\sum_{j=1}^{\tau} C_j(\omega)+\xi(\omega) > 2^{7}\cdot B$ and $\xi(\omega) \leq 2^6\cdot B$, then we know $ \sum_{j=1}^{\tau} C_j(\omega)\geq 2^{6}\cdot B$.
Thus the probability of this Sub-case can be bounded by Markov Inequality.
In particular, one has
\begin{align*}
    \Pr[\omega\in \Omega_{2.1}]
    \leq & \Pr\left[\omega\in\Omega:  \sum_{j=1}^{\tau} C_j(\omega)\geq 2^{6}\cdot B\wedge \sum_{j=1}^{\tau}\E[C_j\mid \omega_{[0:j-1]}] \leq 2^{4}\cdot B\right]\\
    \leq & 
    %\Pr\left[\omega\in \Omega:\sum_{j=1}^{\tau}\E[C_j\mid \omega_{[0:j-1]}] \leq 2^{4}\cdot B\right]\cdot
    \Pr\left[\sum_{j=1}^{\tau} C_j(\omega)\geq 2^{6}\cdot B\mid \sum_{j=1}^{\tau}\E[C_j\mid \omega_{[0:j-1]}] \leq 2^{4}\cdot B \right]
    \leq  2^{-2},
\end{align*}
where the last step follows by applying Markov Inequality on the probability space $\{\omega\in \Omega:\sum_{j=1}^{\tau}\E[C_j\mid \omega_{[0:j-1]}] \leq 2^{4}\cdot B\}$.

{\bf Sub-case (2.3):}[Large expected movement cost, small breaking cost.]
The corresponding sample path set is $\Omega_{2.3}=\{\omega\in \Omega : \T(\omega)<2,\sum_{j=1}^{\tau}X_j(\omega)=0, \sum_{j=1}^{\tau}\E[C_j\mid \omega_{[0:j-1]}] > 2^{4}\cdot B, \sum_{j=1}^{\tau} C_j(\omega)+\xi(\omega) > 2^{7}\cdot B,\xi(\omega)\leq 2^6\cdot B\}$. 

%In particular, we create a new variable $\OPT''$ which is only used in this sub-case.
Let $\M''$ be the instance when the $\tau$-th switch occurs, $k''$ be the remaining number of rewards we need to collect, and let $\OPT''$ be the optimal strategy for $\M''$ with the objective $k''$.
%conditioning on $\omega$. That is, if $\OPT''$ chooses to play the $\S_{\tau}$, then it needs to traverse the $\omega_\tau$.
By the same proof of Lemma~\ref{lm:larger_better} (trajectory-wise), we also have
\[
\E[\TC(\OPT'')]\leq \E[\TC(\OPT)]\leq B/10.
\]

Conditioning on any sample path $\omega\in \Omega_{2.3}$, we know that $\T(\omega)=\sum_{j=1}^{\tau}\E[X_j\mid \omega_{[0:j-1]}]\leq 2$ and $\sum_{j=1}^{\tau}\E[C_j+1\mid \omega_{[0:j-1]}]\geq 2^4\cdot B$. Then there exists $j\leq \tau$ such that $\frac{\E[C_j+1\mid{\omega_{[0:j-1]}}]}{\E[X_j\mid \omega_{[0:j-1]}]}\geq 2^3\cdot B$ and by Claim~\ref{clm:bound_grade}, we know $\Gamma_{i_{j+1}}\geq 2^3\cdot B$.

Denote the subset $\Omega_{2.3.1}=\{\omega\in\Omega_{2.3}:\Gamma_{i_\tau}\geq 2^3\cdot B\}$.
We know that 
\[
\Pr[\OPT'' \text{ succeeds to collect at least one unit reward within Budget $B$}\mid \omega\in\Omega_{2.3.1}]\leq 2^{-3}.
\]
by Claim~\ref{clm:large_index_low_reward}.
Hence, by the fact that $B/10\geq \E[\TC(\OPT'')]\geq \E[\TC(\OPT'')\mid \omega\in \Omega_{2.3.1}]\Pr[\Omega_{2.3.1}] \geq B(1-2^{-3})\Pr[\Omega_{2.3.1}]$,
we can see that $\Pr[\Omega_{2.3.1}]\leq 4/35$ .

Otherwise, consider those $\omega\in\Omega_{2.3}\setminus\Omega_{2.3.1}$ such that $\Gamma_{i_{\tau}}< 2^3\cdot B$.
By the greedy property of the algorithm, we know that $\Gamma_{i_j}\leq \Gamma_{i_{\tau}}$ for all $1\leq j\leq \tau$.
By Claim~\ref{clm:bound_grade}, we know that
$2^3\cdot B\geq \Gamma_{i_{j+1}}\geq   \frac{\E[C_j+1\mid \omega_{[0;j-1]}]}{\E[X_j\mid \omega_{[0:j-1]}] }$ for $1\leq j \leq \tau-1$.
As $\T(\omega)=\sum_{j=1}^{\tau}\E[X_j\mid \omega_{[0:j-1]}]\leq 2$, then we know that $\sum_{j=1}^{\tau-1}\E[C_j+1\mid\omega_{[0:j-1]}]\leq 2^4\cdot B$. 
However, $\sum_{j=1}^\tau C_j(\omega)>2^6 B$ for this case.
Use the same Markov inequality argument as in Sub-case (2.2), one has
\begin{align*}
    \Pr[\omega\in \Omega_{2.3}\setminus\Omega_{2.3.1}]
    \leq & \Pr\left[\omega\in\Omega:  \sum_{j=1}^{\tau-1} C_j(\omega)\geq 2^{6}\cdot B\wedge \sum_{j=1}^{\tau-1}\E[C_j\mid \omega_{[0:j-1]}] \leq 2^{4}\cdot B\right]\\
    \leq & 
%    \Pr\left[\omega\in \Omega:\sum_{j=1}^{\tau-1}\E[C_j\mid \omega_{[0:j-1]}] \leq 2^{4}\cdot B\right]\cdot
    \Pr\left[\sum_{j=1}^{\tau-1} C_j(\omega)\geq 2^{6}\cdot B\mid \sum_{j=1}^{\tau-1}\E[C_j\mid \omega_{[0:j-1]}] \leq 2^{4}\cdot B \right]
    \leq  2^{-2}.
\end{align*}

%For any $\A$ that is in this sub-case, condition  on $\A$. As the expected movement cost of \METAfindone is more than $2^{7}\cdot B$ and $T<2$,  then there exists $j$ such that $\frac{\E[C_j\mid{\cal F}_{j-1}]}{\E[X_j\mid{\cal F}_{j-1}]}\geq 2^9\cdot B$, which means that \METAfindone must have decided to play some \MS whose index is at least $2^{9}\cdot B$ by Lemma~\ref{lem:giitins_one_round_generally_P2}. Using similar argument in Case(1), one also has:

{\bf Union Bound:}
By union bound over {\bf Case (1)} and {\bf Case (2)}, we have that
\begin{align*}
    \Pr[\omega\in \Omega : \METAfindone \text{ fails and } T(\omega)<2] &= \Pr[\Omega_1 \cup \Omega_{2.1}\cup \Omega_{2.2}\cup\Omega_{2.3}]\\
    &\leq \Pr[\Omega_1]+\Pr[\Omega_{2.1}]+\Pr[\Omega_{2.2}]+\Pr[\Omega_{2.3}]\\
    & \leq 1/9+1/4+1/9+ \Pr[\Omega_{2.3}\setminus\Omega_{2.3.1}]+\Pr[\Omega_{2.3.1}]\\
    &\leq 1/9+1/4+1/9+1/4+4/35\\
    &<  17/20.
\end{align*}
This completes the proof.
\end{proof}

The lemma below complements Lemma~\ref{lm:large_ET} and shows that if the summation of conditional expectation is large, then the algorithm should succeed with a constant probability.
\begin{lemma}
\label{lm:good_proerty_large_T}
Suppose $X_1,X_2,\cdots,X_n$ are a sequence of random variables taking values in $\{0,1\}$, and $\F_j=\sigma(X_1,\cdots,X_j)$ is the filtration defined by the sequence. Given any real number $\mu$, if $\sum_{j=1}^{n}\E[X_j\mid{\cal F}_{j-1}]\geq \mu$, then
\begin{align*}
    \Pr\left[\sum_{j=1}^{n}X_j\geq 1\right]\geq 1-e^{-3\mu/8}.
\end{align*}
\end{lemma}

\iffalse
\begin{proof}
We can use Technical Lemma~\ref{lem:concentration} in the Appendix to get similar result, which constructs a sequence of martingale and applies the Freedman's Inequality (Theorem~\ref{thm:FreedmanInequality}).
Here we show another interesting way to prove it. 

%In our situation, the concrete distributions of $X_j$ is limited by the \MS given in the input. 
%We relax this limitation and try to argue that for a stronger, the probability of $\sum_{j=1}^{2^{7}B}X_j\geq 1$ is at least $1-e^{-2}$.

We can transfer the above problem into such a scenario: suppose we are tossing coins randomly with the hope to make no coins head. We can construct coins adaptively as we wish with only one constraint. More specifically, we can make $j$-th coin head with probability $p_j$ after observing the tossing results of the first $j-1$-th coins, and the two constraint is that
i)we can construct at most $2^{7}\cdot B$ coins; 
ii)and $\sum_{j=1}^{2^{7}\cdot B}p_j\geq \mu$.

It suffices to show that the upper bound of the objective (the probability that no coin heads up) in this scenario is no more than $e^{-2}$. 
Fortunately, it is easy to see that the best strategy is non-adaptive, as once there is one coin head, the game is over.
In the domain of non-adaptive algorithm, making $p_j=\frac{\mu}{2^{7}B}$ for each $j\in[2^{7}\cdot B]$ is optimal, which can be proved by induction and elementary calculation. The best objective is $(1-\frac{\mu}{2^{7}\cdot B})^{2^{7}\cdot B}\leq e^{-\mu}$. This completes the proof.
\end{proof}
\fi

The proof of Lemma~\ref{lm:good_proerty_large_T} can be found in Appendix~\ref{sec:prob_toss_coins}.
By Lemma~\ref{lm:large_ET}, one has
\begin{align}
\label{eq:large_ET}
    \Pr\left[\sum_{j=1}^{2^{7}B}X_j\geq 1\right]+\Pr[T\geq 2]\geq & \Pr\left[\sum_{j=1}^{2^{7}B}X_j\geq 1\vee T\geq 2\right]\geq 3/20.
\end{align}

By Lemma~\ref{lm:good_proerty_large_T}, one has
\begin{align}
\label{eq:good_proerty_large_T}
     \Pr\left[\sum_{j=1}^{2^{7}B}X_j\geq 1 \mid T\geq 2\right]\geq 1-e^{-3/4}\geq 1/2.
\end{align}

By combining Equation~\ref{eq:large_ET} and Equation~\ref{eq:good_proerty_large_T}, one can easily show that 
\begin{align*}
 \Pr\left[\sum_{j=1}^{2^{7}B}X_j\geq 1\right] 
    \geq & \Pr\left[\sum_{j=1}^{2^{7}B}X_j\geq 1 \mid T\geq 2\right]\cdot \Pr[T\geq 2]
    \geq  (1-e^{-3/4})\cdot \Pr[T\geq 2] \\
    \geq & (1-e^{-3/4}) \cdot \left(\frac{3}{20}-\Pr\left[\sum_{j=1}^{2^{7}B}X_j\geq 1\right]\right)
    \geq  \frac{1}{2} \cdot \left(\frac{3}{20}-\Pr\left[\sum_{j=1}^{2^{7}B}X_j\geq 1\right]\right),
\end{align*}
which implies that 
\begin{align*}
    \Pr\left[\sum_{j=1}^{2^{7}B}X_j\geq 1\right]\geq 1/20.
\end{align*}
Thus we complete the proof of Lemma~\ref{lm:precondition}.
\end{proofof}

%\begin{proofof}{Theorem~\ref{thm:Find1MS}}
Theorem~\ref{thm:Find1MS} follows from Lemma~\ref{lem:key_sub_process} and Lemma~\ref{lm:precondition} directly.
%\end{proofof}

The simple index-based strategy achieves a constant approximation factor in unit metric space. 
However, there is a simple counter-example on general metric, which
shows \METAfindone may perform arbitrarily bad. See Appendix~\ref{sec:counter_example}
for the example.
This suggests that new techniques are needed for more general metric. This is the focus of the next section.

\section{Markov Game with General Metric}
\label{sec:general}
In this section, we consider \SFindK (Definition~\ref{def:SFindK}).
Note that our requirement is also more general: we need to collect $K$ units of rewards (make $K$ chains reach their targets). 
Our main result is an efficient strategy $\ALGSfindK$ which can achieve a constant factor approximation.

%\begin{restatable}{theorem}{ThmSFindK}
%\label{thm:SFindkMS}
%There is an efficient $O(1)$-approximation algorithm for \SFindK.
%\end{restatable}

$\ALGSfindK$ (See Algorithm~\ref{alg:alg_SfindK} in Appendix) adopts the same doubling framework as $\ALGgeneral$ (Algorithm~\ref{alg:alg_general}), except that we need a more complicated sub-procedure \METASfindK (which approximates the budgeted version).
In particular, \METASfindK should satisfy the precondition of  Lemma~\ref{lem:key_sub_process}.
i.e., the expected total cost of \METASfindK is bounded by $c_2\E[\TC(\OPT(\M,K))]$ for some universal constant $c_2$, and it can collect $K$ units of rewards with constant probability when $B$ is large enough. 
For clarity, we present the lemma below, which is the precondition of Lemma~\ref{lem:key_sub_process} specialized for $\SFindK$.

\begin{lemma}
\label{lem:modified_subprocess_simple_FindK_metric}
For any input $\M$ with the objective number of rewards $K$, let $\OPT$ be the optimal strategy for this instance. If $B\geq10\E[\TC(\OPT)]$, with probability at least 0.1, $\METASfindK$ (Algorithm~\ref{alg:meta_SfindK}) can collect at least $K$ units of rewards (i.e. make at least $K$ system reach their targets) with budget $O(1)B$ in expectation.
\end{lemma}

\begin{comment}
\begin{lemma}
\label{lem:modified_subprocess_simple_FindK_metric}
    We are given some \SFindK instance $\M$, a positive integer $K$
    which is the number of rewards we need to collect, and a non-negative real number $ B\in\mathbb{R}_{\geq 0}$. 
    %which is the total budget. 
    Let $\OPT$ be the optimal strategy of $\M$ 
    (the original Markov game).
    Suppose there is an algorithm $\METASfindK$
    that can make $K$ chains reach their targets with probability at least 0.1 for $B\geq  c_1\E[\TC(\OPT)]$ and satisfies that $c_2B\geq \E[\TC(\METASfindK)]$, where $c_1,c_2$ are some universal constants, then $\ALGSfindK$ achieves a constant approximation factor for \SFindK.
\end{lemma}
\end{comment}

It suffices to prove that sub-procedure \METASfindK (Algorithm~\ref{alg:meta_SfindK}) satisfies this lemma, which is the main task of this section.
\subsection{Stochastic $k$-TSP.}
%In this subsection, we analyse \METASfindKfindK and show it satisfies the precondition of Lemma~\ref{lem:modified_subprocess_simple_FindK_metric}.
The sub-procedure \METASfindK makes use of an $O(1)$-approximation  for the \StochKCost problem \cite{ENS18,JLLS20}, defined as follows.

\begin{definition}[\StochKCost]
We are given a metric $\M=(\mathbf{S}, d)$ with a $\depot \in \mathbf{S}$ and each vertex $v \in \mathbf{S}$ has an independent stochastic selection cost $\CSC(v) \in \mathbb{R}_{\geq 0}$. All cost distributions are given as input but the actual cost instantiation $\CSC(v)$ is only known after vertex $v$ is visited. 
Suppose a vertex $v$ can only be \emph{selected} if $v$ is visited.\footnote{\label{footnote:DistSelect} 
In the original version in \cite{JLLS20}, there is one more condition that we are currently at vertex $v$ to select it.
Up to a factor of 2, our version of the problem is equivalent to the original version.} 
The goal is to adaptively find a tour $T$ originating from the  $\depot$ and select a set $S$ of $k$ visited vertices while minimizing the expected {\em total} cost, which is the sum of the length of $T$ and the cost of the selected vertices:
\[
\E \Big[\mathrm{Length}(T) + \sum_{v \in S} \CSC(v) \Big].
\]
\end{definition}

\begin{theorem}[Theorem 2 in \cite{JLLS20}]
\label{thm:const_scktsp}
There is a non-adaptive constant factor approximation algorithm $\ALGStochKCost$ for \StochKCost.
\end{theorem}
\begin{remark}
It is important to note that the strategy in \cite{JLLS20}
is non-adaptive.
Here a non-adaptive strategy is an ordering $\Pi$ of all vertices and $\Pi$ is independent of the realization of the costs.
Note that the strategy visits the vertices according to the ordering $\Pi$ and may stop and choose $k$ visited vertices before visiting all vertices by some criterion, depending on the realization of the costs.
We use $\PG =O(1)$ to be the constant approximation ratio of this algorithm.
\end{remark}

%We abuse the notation.
As mentioned, $\ALGStochKCost$ is not only a sub-procedure which can output an ordering of the chains (vertices). In fact there is also some probing and selection process after outputting the ordering in $\ALGStochKCost$.
We let $\CSC(\ALGStochKCost)$ and $\SC(\ALGStochKCost)$ be the (random) selection cost and switching cost of $\ALGStochKCost$ respectively.

%\begin{lemma}[Modified Lemma 13 in \cite{JLLS20}]
%Given any \SFindK instance, non-negative real number $B\in\mathbb{R}_{\geq 0}$. Suppose there exists a strategy \OPT, which with probability at least 0.9 can make at least $K$ \MS reach their target with switching budget $B$ and movement budget $B$. Then there is an efficient algorithm \ALGStochKCost which can output a path $\Pi$ of candidate vertices subject to (1) the length of $\Pi$ is at most $3B$, (2) with probability at least 0.2, we can select at least $K$ vertices in $\Pi$ with cost $6B$.
%\end{lemma}

%\METASfindK will transfer the \SFindK instance to a \StochKCost problem and use \ALGStochKCost as a black box to get a set of \MS at first. More specifically, for each $\S=\langle V,P,C,s,t \rangle$, we consider the game $\S^T$. Let $C_{\S}$ represent the reward the non-quitting player get from $\S^T$. Then we treat the $\S$ as a vertex in \StochKCost and let $C_{\S}$ be its cost and get a \StochKCost instance $\M_{ktsp}$. The distribution of $C_{\S}$ can be calculated. Particularly, we can solve a linear system to calculate $\Pr[C_{\S}\leq C]$ for any real number $C$. Then \METAfindK will use \ALGStochKCost to solve $\M_{ktsp}$ and accept the set of vertices outputted by \ALGStochKCost.

\subsection{Reduction to \StochKCost}
\subsubsection{A Fair Game \SFindKF.}
\label{sec:fair_game}
We define another new game which is closely related to the teasing game $\S^T$ (see definition~\ref{def:teasing_game}) and plays a key role in the following proof.
\begin{definition}[\SFindKF]
We are given a finite metric space $\M=(\mathbf{S}\cup\{\mathbf{R}\},\dist)$ (there is no additional assumption on metric $\dist$). 
Each node $\S_i\in \mathbf{S}$ is identified with a \MS 
$\S_i=\langle V_i,P_i,C_i,s_i,t_i \rangle$. 
Similarly, at the beginning of the game, the player is at the root $\mathbf{R}$, and needs to pay the switching cost $\dist(\mathbf{R},\S_i)$ if he wants to play \MS $\S_i$. 
Switching from $\S_i$ to $\S_j$ incurs a switching cost of  $\dist(\S_i,\S_j)$.
Let $V_i^F\subseteq V_i$ denote the subset of states $\S_i$ has been transited to by the player during the game.
If $t_i\in V_i^F$, he can pay the prevailing cost as the fair movement cost and get one reward from $\S_i$.
The objective is to adaptively collect at least $K$ units of rewards (making at least $K$ \MS reach their targets), while minimizing the expected total cost (fair movement cost plus switching cost).
\end{definition}
 
Let  $\FMC(\P)$, $\SC(\P)$ and $\Ftotalcost(\P)$ be the (random) fair movement cost, switching cost and total cost of any strategy $\P$ under the rule of \SFindKF respectively. 
Let $\FOPT$ be the optimal strategy to \SFindKF.
We have the following claim. The proof can be found in Appendix~\ref{sec:omitted proof}.

\begin{claim}
\label{clm:relation_two_games}
The following inequality holds:
$
    \E[\Ftotalcost(\P)]=\E[\SC(\FOPT)+\FMC(\FOPT)]\leq \E[\TC(\OPT)].
$
\end{claim}
%This claim will be used to prove the Claim~\ref{clm:games_reduction}.

\subsubsection{The reduction.}
\label{subsec:reduction}
At a high level, $\METASfindK$ reduces an instance of \SFindK 
to an instance of \StochKCost which has the same metric.
A Markov chain $\S$ in \SFindK is replaced by a random variable, i.e. the random selection cost.
More specifically, for each $\S=\langle V,P,C,s,t \rangle$, we consider the teasing game $\S^T$. 
Let $\D(\S)$ be the (random) prevailing cost of a non-quitting strategy
%gets under the rule of $\S^T$ (See definition~\ref{def:teasing_game}), 
where being non-quitting means that the player continues to play $\S^T$ until it reaches the target state. 
Then we treat $\S$ as a vertex in \StochKCost and let its selection cost distributed as $\D(\S)$. 

We note that the distribution of $\D(\S)$
(in particular, $\Pr[\D(\S)\leq B]$ for any real number $B$) 
can be computed efficiently
using the technique for computing grade
(see Lemma~\ref{lm:cal_cdf} in Appendix~\ref{appd:computation_cdf}). 

We establish a simple relation between 
the optimal cost of \StochKCost and optimal cost of our original problem \SFindK.
We let $\E[\CSC(\P)]$, $\E[\SC(\P)]$, $\E[\TSPCOST(\P)]$ be the expected selection cost, expected switching cost and 
expected total cost of any strategy $\P$ for the \StochKCost problem, respectively. 
Let $\SCOPT$ be the optimal strategy to \StochKCost. 
Note that we use the same notation $\SC(\cdot)$ to represent the switching cost for all games we have defined since the switching costs are counted in the same way. 

Given any instance $\M$, let $\OPT$ be the optimal strategies for this instance $\M$ for our original problem \SFindK.
Now, we claim that 
the optimal cost of \StochKCost is no more than the optimal cost of our original problem. 
The proof of Claim~\ref{clm:games_reduction} is via the fair game we introduced in Section ~\ref{sec:fair_game}
and can be found in Appendix~\ref{sec:omitted proof}.

\begin{claim}
\label{clm:games_reduction}
With the notations defined above, it holds that
\begin{align*}
    \E[\TSPCOST(\SCOPT)]=\E[\SC(\SCOPT)+\CSC(\SCOPT)]\leq \E[\TC(\OPT)].
\end{align*}
\end{claim}
Recall we are under the assumption that $\E[\TC(\OPT)]\leq B/10$. 
Then we have
\begin{align*}
    \E[\TSPCOST(\ALGStochKCost)]\leq & \PG \E[\TSPCOST(\SCOPT)]\\
    \leq & \PG \E[\TC(\OPT)] \\
    \leq & \PG B/10,
\end{align*}
where the first line is by Theorem~\ref{thm:const_scktsp} and the second line follows from Claim~\ref{clm:games_reduction}.
%Then applying this claim and the Theorem~\ref{thm:const_scktsp}, we can bound that the expected cost of $\ALGStochKCost$ by $\PG B/10$.

\subsection{Sub-procedure $\METASfindK$.}
Now we provide a high level description of $\METASfindK$.
The details can be found in Algorithm~\ref{alg:meta_SfindK}.
We first transform the problem to a \StochKCost instance $\M_{ktsp}$,
by reducing each Markov chain to a related random variable.
Then we use the constant factor approximation algorithm 
$\ALGStochKCost$ developed in \cite{JLLS20} to obtain an ordering
$\Pi$ of vertices (chains).
Let $\Pipref$ be the prefix of $\Pi$ such that the switching cost
for traversing $\Pipref$ is no larger than $10\PG B$.
Since the expected cost of the
$\ALGStochKCost$ of instance $\M_{ktsp}$ is at most $\PG B/10$, with a large constant probability, one can collect $K$ units of rewards from $\Pipref$, by Markov inequality.
Obviously, ignoring the switching cost, the optimal way 
(optimal in terms of movement cost) of collecting $K$ units of rewards from $\Pipref$ is to play the chains in $\Pipref$ 
according to grade. However, such play may switch frequently among the chains and incur a huge switching cost.
To keep the switching cost under control, we visit each chain
in $\Pipref$ only once (by the way we define the prefix, we can see the switching cost is certainly within the given budget). 

A key question now is to decide when 
to switch to the next if the current state is not economical to keep playing (it requires a large expected movement cost to reach the target of the chain). 
It turns out that we can find a threshold $\gamma_{j+1}$ (which is computed from the $K$-th order statistics
of $\{\D(\S)\}_{\S\in \Pipref}$), such that 
if the grade of the current state is larger than the threshold,  
or our algorithm has spent too much on this chain, say a movement budget $100\PG B$ for each chain,
then it is time to switch to the next chain on the  $\Pipref$.

As mentioned, our algorithm needs to estimate the distribution of the $K$-th {\em order statistic} for a collection of random variables
(that is the $K$-th smallest value).
This can be approximated simply either by sampling (Monte Carlo) or 
the Bapat-Beg theorem and the fully polynomial randomized approximation scheme (FPRAS) for estimating the permanent \cite{JSV04} (see the details in Appendix~\ref{sec:est_order_stat}).
In the following, we safely ignore the estimation error and error probability for clarity purpose.

\begin{algorithm2e}[ht]
\caption{Algorithm  $\METASfindK$}
\label{alg:meta_SfindK}
{\bf  Input:} The instance $\M$, objective number of rewards $K$\\
%budget $50000 \PG  B$
{\bf Process:}\\
%Set $k\leftarrow K$\;

Compute the grades of all states and sort them in increasing order $\{\gamma_j\}_{j=1}^{n}$\;

%Compute the the distribution of $\D(\S)$ for all Markov chains $\S$ 

Reduce the instance $\M$ to a \StochKCost instance $\M_{ktsp}$, by replacing each $\S\in \M$ with a vertex in $\M_{ktsp}$ with selection
cost distributed as $\D(\S)$
(See the definition of $\D(\S)$ in Section~\ref{subsec:reduction})\;

$\Pi  \leftarrow  \ALGStochKCost(\M_{ktsp})$ \label{lin:good_set}
($\Pi$ is an ordering of vertices)
\;

Let $\Pipref=\{\S_0=\mathbf{R},\S_1,\S_2,\cdots,\S_m\}$ be the longest prefix of $\Pi$ such that $\sum_{i=0}^{m-1}d(\S_i,\S_{i+1})\leq 10\PG B$\;

Estimate the distribution of $K$-th order statistic $\D^{[K]}(\S)$ for $\{\D(\S)\}_{\S\in \Pipref}$\;

Find the unique $\gamma_j$, $j\in [n]$ such that
$\Pr[\D^{[K]}(\S)\leq \gamma_j]< 0.3\wedge \Pr[\D^{[K]}(\S)\leq \gamma_{j+1}]\geq 0.3$ \label{lin:threshold}\;
\For{ $i=1,\cdots,m$}
{
Pay switching cost $d(\S_{i-1},\S_i)$ and switches to $\S_i$\;

Play $\S_i$ until it reaches its target state or some condition(s) in $\mathcal{A}$ (defined below) holds.
\label{lin:single_OPT}\;

If $\S_i$ reaches its target state, let $K\leftarrow K-1$\;

%If condition (I) in $\mathcal{A}$ (defined below) holds, break;
}

{\bf Return:} The updated instance $\M$; the remaining number of target states $K$\;
\BlankLine\;
{\bf Define:} Conditions $\mathcal{A}$ (bad events):\;
%(I) The next move on $\S_i$ make total movement cost exceed $50000\PG B$\;
(I) The next move on $\S_i$ makes the movement cost on $\S_i$ exceed $100\PG B$\;
%(III) And $\S_i$ is not at target state $t_i$\;
(II) The current grade of $\S_i$ is more than $\gamma_{j}$\;
\end{algorithm2e}

%is negligible and the error probability $1-1/\poly(n)$ which can be absorbed into the failure probability of \METASfindK.
%For simplicity, we omit the error and failure on the calculation on order statistic which are too minor to influence the correctness of our proof.

\subsection{Analysis.}
Recall that it suffices to prove that \METASfindK satisfies Lemma~\ref{lem:modified_subprocess_simple_FindK_metric}.
We only need to consider the case $B\geq 10\E[\TC(\OPT)]$, and prove two guarantees of \METASfindK:\\
i) It can succeed to make $K$ \MS reach their target states with probability at least 0.1;\\
ii) The expectation of its total cost is bounded by $400\PG B$.

%Recall that $\Pr[\D^{[k]}(\S)\leq \gamma_{j+1}]\geq 0.3$ and if we run the algorithm with stopping Conditions (I) and (II) removed, the success probability is at least 0.3. Then we consider the influence of Condition (I) and Condition (II) separately. For Condition (I), we have the following result:

Recall that $\D(\S)$ represents the (random) selection cost of $\S$ under the rule of \StochKCost.
The set $\Pipref$ we found in the line~\ref{lin:good_set} of Algorithm~\ref{alg:meta_SfindK} has some good properties stated below:
\begin{lemma}
\label{lem:good_set_guarantee}
Suppose $B\geq 10\E[\TC(\OPT)]$. The set $\Pipref$ found in the line~\ref{lin:good_set} of Algorithm~\ref{alg:meta_SfindK} satisfies that with probability at least 0.99, one has $\sum_{i=1}^{K}\D^{[i]}(\S)\leq 10\PG B$, where $\D^{[i]}(\S)$ is the $i$-th order statistic for $\{\D(\S)\}_{\S\in \Pipref}$.
\end{lemma}

\begin{proof}
As $B\geq 10\E[\TC(\OPT)]$, we know that $B\geq 10\E[\CSC(\SCOPT)+\SC(\SCOPT)]$ by Claim~\ref{clm:games_reduction}. 

By Theorem~\ref{thm:const_scktsp}, we know that
$\PG B\geq 10\E[\CSC(\ALGStochKCost)+\SC(\ALGStochKCost)]$.
By Markov Inequality, we know that
\begin{align}
\label{eq:tail_ktsp}
    \Pr[\CSC(\ALGStochKCost)+\SC(\ALGStochKCost)\geq 10\PG B]\leq 0.01.
\end{align}

If $\sum_{i=1}^{K}\D^{[i]}(\S)> 10\PG B$, then it means that
\begin{align}
\label{eq:contradition}
    \CSC(\ALGStochKCost)+\SC(\ALGStochKCost)\geq 10\PG B.
\end{align}
More specifically, if $\CSC(\ALGStochKCost)\leq 10\PG B$, as $\sum_{i=1}^{K}\D^{[i]}(\S)> 10\PG B$, then $\ALGStochKCost$ needs to visit vertices outside $\Pipref$ to select some ``cheap'' vertices, which incurs a switching cost larger than $10\PG B$.
So either switching cost or the selection cost of $\ALGStochKCost$ is larger than $10\PG B$.

By combining Equation~\ref{eq:tail_ktsp} and Equation~\ref{eq:contradition}, we get
\begin{align*}
    \Pr\left[\sum_{i=1}^{K}\D^{[i]}(\S)\leq 10\PG B\right]\geq 0.99.
\end{align*}
%which implies that $\PG B/10\geq \E[\CSC(\ALGStochKCost)+\SC(\ALGStochKCost)]$ and if we take the path $\Pi$ of vertices $\ALGStochKCost$ choose to visit with cost $10\PG B$, then with probability at least $0.99$, we can select at least $K$ vertices in $\Pi$ with a total cost no more than $10\PG B$, i.e. we have the following guarantee;

\end{proof}

%Suppose we take $\Pi$ as an instance of \SFindKNM (no switching cost) and define the strategy $\ALGGITTINS$ as following: always choose the \MS with the smallest Provisional Index to play, until making $K$ systems reach their targets or the summation of paid cost plus the smallest Provisional Index exceeds $100\PG B$. 
%Recall that we use $\M=\{\mathsf{S}\cup\{\mathbf{R}\},d\}$ to denote the initial input of $\METASfindK$.
%Then let $\M_{10\PG B}=\{ \mathsf{S}\cup\{\mathbf{R}\} \cap \Pipref,d\}$ be a new instance of \SFindKF we will consider.

We make use of the good properties of $\Pipref$ via the fair game we define in Section ~\ref{sec:fair_game}.
Define an important intermediate strategy 
$\ALGGITTINS$ for the input istance $\M$ under the rule of $\SFindKF$ as follows:
it always chooses the \MS in $\Pipref$
with the smallest grade  (breaking ties arbitrarily) to play under the condition that $\FMC(\ALGGITTINS)\leq 10\PG B$, pays the {\em fair} movement cost {\em immediately} when some chain reaches its target and halts when it makes $K$ systems reach their targets or the next step would break the fair movement cost budget.

To be more clear on the stopping condition of $\ALGGITTINS$ on the cost budget, suppose $\ALGGITTINS$ has not made $K$ chains reach targets, the fair movement cost spent is $X$ and the smallest grade among available chains (those not in the target states) is $Y$. Then $\ALGGITTINS$ chooses the chain with smallest grade $Y$ if $X+Y\leq 10\PG B$, and halts if no such chain exists.\footnote{In fact, if there is no switching cost and no movement cost budget, $\ALGGITTINS$ is the optimal solution to make $K$ chains in $\Pipref$ to reach targets with the minimum expected movement cost.}
Then we have the following claim:

\begin{claim}
\label{clm:guarantee_on_OPTCB}
With probability at least 0.99, $\ALGGITTINS$ can succeed to make $K$ systems reach their targets.
Further, one has:
\label{clm:low_expected_cost}
\[
\E[\MC(\ALGGITTINS)]\leq 10\PG B.
\]
\end{claim}

\begin{proof}
The probability of success is a direct corollary of Lemma~\ref{lem:good_set_guarantee}, as the fair movement cost of a non-quitting player on single $\S$ is exactly distributed as the selection cost $\D(\S)$.
By Lemma~\ref{lem:good_set_guarantee},
one has
\begin{align*}
    \Pr\left[\sum_{i=1}^{K}\D^{[i]}(\S)\leq 10\PG B\right]\geq 0.99,
\end{align*}
where $\D^{[i]}(\S)$ is the $i$-th order statistic for $\{\D(\S)\}_{\S\in \Pipref}$.
Suppose there are $m$ non-quitting players and every one plays one chain on $\Pipref$. With probability at least 0.99, the summation of the smallest $k$ fair movement costs they pay will be no more than $10\PG B$,
in which case $\ALGGITTINS$ can succeed to make $K$ systems reach their targets.

The bound on movement cost follows from the definition and the relation that
\begin{align*}
    \E[\MC(\ALGGITTINS)]= \E[\FMC(\ALGGITTINS)]\leq 10\PG B,
\end{align*}
where the first equality comes from that $\ALGGITTINS$ is a fair player under the rule of \SFindKF.
More specifically, we can observe that $\ALGGITTINS$ is playing a series of teasing game $\S^T$, and it can get the profits under the rule of $\S^T$ and need to pay the fair movement cost under the rule of \SFindKF instead. 
Moreover, its ``expected fair movement cost'' equals to ``expected profits''.

As it plays all of teasing games optimally by Lemma~\ref{lm:fairness}, the ``expected profits'' equals to its ``expected movement cost''.
Thus, we get the equality.
\end{proof}

Recall again the objective is to show that $\METASfindK$ satisfies the precondition of Lemma~\ref{lem:modified_subprocess_simple_FindK_metric}. 
%more specifically, the constant success probability and bounded expected total cost.
We prove the guarantee on success probability first:
\begin{lemma}[Success probability]
With probability at least 0.1, \METASfindK can make at least $K$ \MS reach their target states.
\end{lemma}
\begin{proof}
Note that we use $\D(\S)$ for system $\S$ to represent the random selection cost under the rule of $\StochKCost$, and now we consider the instance under the original rule, i.e. in \SFindK.

For simplicity, we use $\Z(\S,\gamma_{j+1})$ to represent the (random) movement cost $\MC (\OPT(\S(\gamma_{j+1})))$, where the game $\S(\gamma_{j+1})$ and its optimal strategy $\OPT(\S(\gamma_{j+1}))$ is used to defined the grade in Subsection~\ref{sec:gittins_index}.

Suppose we play each $\S\in\Pipref$ one by one according to $\OPT(\S(\gamma_{j+1}))$ and denote the random subset of systems which reach their targets by $\Pi_{\mathrm{suc}}\subseteq \Pipref$.
If $|\Pi_{\mathrm{suc}}|\geq k$, we let
$\Z^{[i]}$ be the $i$-th order statistic for $\{\Z(\S)\}_{\S\in \Pi_{\mathrm{suc}}}$ and $i\in[k]$, otherwise we define $\Z^{[i]}$ to be the $i$-th order statistic when $i\in[|\Pi_{\mathrm{suc}}|]$ and $\Z^{[i]}=0$ for $ |\Pi_{\mathrm{suc}}|<i\leq k$.
We have the following claim:
\begin{equation}
\label{eq:bound_cost}
\Pr[\sum_{i=1}^{K}\Z^{[i]} \geq 100\PG B]\leq 0.2.
\end{equation}

This can be proved by contradiction. 
Suppose $\Pr[\sum_{i=1}^{K}\Z^{[i]} \geq 100\PG B]> 0.2$. Since the success probability of $\ALGGITTINS$ is at least $0.99$, by union bound, one can see that
\begin{align*}
 \Pr\left[\sum_{i=1}^{K}\Z^{[i]} \geq 100\PG B\wedge\ALGGITTINS \text{ gets } K \text{ rewards } \right] \geq &~ 0.19.
\end{align*}
Further, we know when event that $(\sum_{i=1}^{K}\Z^{[i]} \geq 100\PG B\wedge\ALGGITTINS \text{ gets } K \text{ rewards })$ happens, then the event $(\MC(\ALGGITTINS)\geq 100\PG B\wedge\ALGGITTINS \text{ gets } K \text{ rewards })$ also happens, as $\sum_{i=1}^{K}\Z^{[i]}$ is the minimum possible movement cost of any strategy which makes $K$ systems on $\Pipref$ reach targets.
This implies that
\begin{align*}
    \Pr\left[\MC(\ALGGITTINS)\geq 100\PG B\wedge\ALGGITTINS \text{ gets } K \text{ rewards }\right] 
    \geq~\Pr\left[\sum_{i=1}^{K}\Z^{[i]} \geq 100\PG B\wedge\ALGGITTINS \text{ gets } K \text{ rewards } \right]
    \geq 0.19.
\end{align*}

Further, the inequality $\Pr[\MC(\ALGGITTINS)\geq 100\PG B]\geq 0.19$ follows directly and
the expected movement cost of $\ALGGITTINS$ is at least $100\PG B*0.19=19\PG B$.
This contradicts Claim~\ref{clm:guarantee_on_OPTCB}.
Thus we complete the proof of Equation~\ref{eq:bound_cost}.

Moreover, we know that $\Pr[\D^{[k]}(\S)\leq \gamma_{j+1}]\geq 0.3$ by the condition in Line~\ref{lin:threshold} in \METASfindK.
By Union bound, we know that
\begin{align*}
    \Pr\left[\D^{[k]}(\S)\leq \gamma_{j+1}\wedge \sum_{i=1}^{K}\Z^{[i]} \leq 100\PG B\right]\geq 0.1,
\end{align*}
which implies the success probability of \METASfindK.
More specifically, the event $\D^{[k]}(\S)\leq \gamma_{j+1}$ ensures that the random subset $|\Pi_{\mathrm{suc}}|\geq K$, in which case event $\sum_{i=1}^{K}\Z^{[i]} \leq 100\PG B$ implies that our algorithm $\METASfindK$ can make at least $K$ chains reach target states.
\end{proof}

Now it remains to bound the expected total cost of \METASfindK.
We state a technical result at first.
%whose proof can be found in Appendix~\ref{sec:omitted proof}.
%Now we have proved the success probability of our algorithm. Another thing we need to do is to upper bound its expected total cost.  As the length of $\Pi$ is at most $O(1)B$ by Lemma~\ref{lem:good_set_guarantee}, we only need to consider the expected movement cost.
%Now we remove the second constraint in condition ${\cal A}$, keep the remaining unchanged and get a new algorithm $\METASfindK'$. We consider the $\METASfindK'$ as a bridge to prove the good guarantees of $\METASfindK$.

\begin{claim}
\label{clm:high_pro_cost}
With probability at least 0.65, $\ALGGITTINS$ can succeed to make $K$ systems reach their targets with no less movement cost than $\METASfindK$.
In other word, one has
\begin{align*}
    \Pr[\ALGGITTINS \text{ gets } K \text{ rewards } \wedge \MC(\ALGGITTINS)\geq\MC(\METASfindK)]\geq 0.65.
\end{align*}
\end{claim}
\begin{proof}
On one hand, by Claim~\ref{clm:guarantee_on_OPTCB}, we know the success probability of $\ALGGITTINS$ is at least 0.99, i.e.
\begin{align}
\label{eq:success_probability}
    \Pr[\ALGGITTINS \text{ gets } K \text{ rewards }]\geq 0.99.
\end{align}

On the other hand, we know that $\Pr[\D^{[K]}(\S)>\gamma_{j}]\geq 0.7$, which means that
\begin{align}
\label{eq:k-order statistic}
    \Pr[\D^{[K]}(\S)\ge \gamma_{j+1}]\geq 0.7,
\end{align}
as we have assumed that the grade is distinct of each other. 
By Union Bound over Equation~\ref{eq:success_probability} and Equation~\ref{eq:k-order statistic}, one has
\begin{align*}
    \Pr[\ALGGITTINS \text{ gets } K \text{ rewards }\wedge\D^{[K]}(\S)\ge \gamma_{j+1}]\geq 0.69.
\end{align*}

Conditioning on that $\D^{[K]}(\S)\ge \gamma_{j+1}$ holds and $\ALGGITTINS$ succeeds to get $K$ units of rewards, the grade of all $\MS\in\Pipref$ after $\ALGGITTINS$ halts are at least $\gamma_{j+1}$.
This is true as $\ALGGITTINS$ always chooses the $\MS$ in the $\Pipref$ with the smallest grade to advance.
As $\METASfindK$ never plays the system when the grade of its current state exceeds $\gamma_{j+1}$,
we know $\MC(\ALGGITTINS)\geq \MC(\METASfindK)$.
In particular,
\begin{align*}
    \Pr[\MC(\ALGGITTINS)\geq \MC(\METASfindK)]
    \geq \Pr[\ALGGITTINS \text{ gets } K \text{ rewards }\wedge\D^{[K]}(\S)\ge \gamma_{j+1}]
    \geq  0.69.
\end{align*}

Combining these together, one has
\begin{align*}
   &~ \Pr[\ALGGITTINS \text{ gets } K \text{ rewards } \wedge \MC(\ALGGITTINS)\geq\MC(\METASfindK)]\\
   = & ~ \Pr\left[\MC(\ALGGITTINS)\geq\MC(\METASfindK)\mid \ALGGITTINS \text{ gets } K \text{ rewards}\right]\cdot \Pr[\ALGGITTINS \text{ gets } K \text{ rewards}]\\
   \geq & ~ \Pr\left[\D^{[K]}(\S)\ge \gamma_{j+1}\mid \ALGGITTINS \text{ gets } K \text{ rewards}\right]\cdot \Pr[\ALGGITTINS \text{ gets } K \text{ rewards}]\\
   \geq &  \Pr\left[\D^{[K]}(\S)\ge \gamma_{j+1}\mid \ALGGITTINS \text{ gets } K \text{ rewards}\right] * 0.99,
\end{align*}
where the first inequality comes from the argument above and the second inequality comes from Equation~\ref{eq:success_probability}.

The remaining is to bound $\Pr\left[\D^{[K]}(\S)\ge \gamma_{j+1}\mid \ALGGITTINS \text{ gets } K \text{ rewards}\right]$.
By conditional probability, we know that
\begin{align*}
     \Pr\left[\D^{[K]}(\S)\ge \gamma_{j+1}\mid \ALGGITTINS \text{ gets } K \text{ rewards}\right]
    =  \frac{\Pr[\D^{[K]}(\S)\ge \gamma_{j+1}\wedge \ALGGITTINS \text{ gets } K \text{ rewards}]}{\Pr[\ALGGITTINS \text{ gets } K \text{ rewards}]}
    \geq  0.69,
\end{align*}
which can complete the proof with elementary calculation.
\end{proof}

Now we are ready to bound the total cost.

\begin{lemma}[Bounded expected cost]
\label{lem:bounded_expected_cost}
The expected total cost of \METASfindK is upper bounded by $O(1)B$.
More specifically, one has
\begin{align*}
   \E[\TC(\METASfindK)]\leq 410\PG B.
\end{align*}
%Let $X(\S)$ be the random movement cost $\METASfindK$ spends on $\S\in \Pipref$.One has
%\[\sum_{\S\in\Pipref}\E[\MC(\OPT_{\S,u,\gamma_{j+1},100\PG B})]\leq 500\PG B.\]
\end{lemma}

\begin{proof}
We prove this by contradiction.
Suppose $\E[\TC(\METASfindK)]> 410\PG B$.
Recall that the switching cost of \METASfindK is at most $10\PG B$. 
Then the assumption implies that
\begin{align*}
    \E[\TC(\METASfindK)]>400 \PG B.
\end{align*}
In fact, $\METASfindK$ plays each system in $\Pipref$ independently of other system.
Recall that we let $\Z(\S)$ be the (random) movement cost $\MC (\OPT(\S(\gamma_{j+1})))$.
%where $\OPT(\S(\gamma_{j+1}))$ is defined in ???.
The only difference between $\METASfindK$ and $\OPT(\S(\gamma_{j+1}))$ on $\S$ is that $\METASfindK$ stops playing $\S$ if the next step will break the movement budget $100\PG B$.
We let $\Z'(\S)$ be the (random) movement cost $\METASfindK$ spends on $\S$ and obviously, we have $\E[\MC(\METASfindK)]=\sum_{\S\in\Pipref}\E[\Z'(\S)]$.

The assumption implies that $\sum_{\S\in\Pipref}\E[\Z'(\S)]> 400\PG B$. 
By the Condition (I) in ${\cal A}$, $\Z'(\S)\in [0,100\PG B]$. 
By Chernoff bound (Theorem~\ref{thm:ChernoffBound}), one has:
\begin{align*}
    \Pr\left[\sum_{\S\in\Pipref}\Z'(\S)\leq 200\PG B\right]
    \leq  \Pr\left[\bigg|\sum_{\S\in\Pipref}\Z'(\S)-\E[\sum_{\S\in\Pipref}\Z'(\S)]\bigg|\geq 300\PG B\right]
    \leq  2e^{-9/2}
    \leq  0.1,
\end{align*}
which means with probability at least 0.9, the movement cost of $\METASfindK$ exceeds $200\PG B$, i.e.
\begin{align}
    \Pr[\MC(\METASfindK)\geq 200\PG B]=\Pr\left[\sum_{\S\in\Pipref}\Z'(\S)\geq 200\PG B\right]\geq 0.9.
\end{align}

In particular, by Claim~\ref{clm:high_pro_cost}, with probability at least 0.65, $\ALGGITTINS$ pays more movement cost than $\METASfindK$. 
Then by union bound, one has
\begin{align*}
    \Pr[ \MC(\ALGGITTINS)\geq 200\PG B]\geq 0.55, 
\end{align*}
which implies that $\E[\MC(\ALGGITTINS)]\geq 110\PG B$ and contradicts Claim~\ref{clm:low_expected_cost}.

Thus we complete the proof.
\end{proof}

Hitherto we have proved that \METASfindK satisfies Lemma~\ref{lem:modified_subprocess_simple_FindK_metric}, and thus finished the proof of the main result:
\FindKMS*

\section{Conclusions and Open Problems}
In this work, we 
present a simple index strategy for \FindOne and a more involved algorithm for \SFindK, both achieving constant approximation ratios.
We did not attempt to optimize the exact constants and the constants directly implied from our analysis are quite large for both problems.
\footnote{The approximation ratio implied from the current analysis of \FindOne is around $10^5$, estimated as follows.
\METAfindone uses $2^8B$ budget and can succeed with probability at least 1/20 conditioning on that the expected cost of optimal solution is no more than $B/10$.
Plugging these constants in Lemma~\ref{lem:key_sub_process}, we know the $\beta$ in the doubling framework should be set as $1/\beta^2=1-1/20=0.95$, and the final approximation ratio can be around $2^8*\frac{10}{\beta-1}\approx 10^5$.
The approximation ratio of \SFindK depends on the the approximation ratio for 
\StochKCost, which is already quite large.
}
Designing new algorithms or analysis with small approximation constants is a very interesting further direction. In particular, we suspect the approximation ratio of \FindOne is a small constant.

One interesting future direction is to study the general problem proposed in
\cite{GJSS19} with switching cost. In particular, there is a given combinatorial constraint $\calF\subseteq 2^{[n]}$, and our goal is to make a subset $F\in \calF$
of chains reach their targets.
Another interesting extension is to significantly generalize the stochastic reward $k$-TSP studied in \cite{ENS18,JLLS20} as follows: we have metric graph, in which each node is associated with a Markov chain. Each target state of a Markov chain
is associated with a random reward $R_v\in \mathbb{Z}^+$, which is realized when we reach the target. The goal is to collect a total reward of at least $k$.

\section{Acknowledgments}
The research is supported in part by the National Natural Science Foundation of China Grant 61822203, 61772297, 61632016, 61761146003, and Turing AI Institute of Nanjing and Xi'an Institute for Interdisciplinary Information Core Technology,
NSF awards CCF-1749609, DMS-1839116, DMS-2023166, CCF-2105772, Microsoft Research Faculty Fellowship, Sloan Research Fellowship, and Packard Fellowships.
Part of the work was done while Liu visited Shanghai Qi Zhi Institute.

%\newpage
%\input{new_general_findk}

%\newpage
%\input{General_Find_K}

\newpage
\appendix
\section{Appendix for Section~\ref{sec:prel}}
\subsection{Concentration Inequalities}
\label{sec:concentration}
We need the following two well known 
concentration inequalities.

\begin{theorem}
[Chernoff-Hoeffding Bound]
\label{thm:ChernoffBound}
Let $X_1,X_2,\cdots,X_n$ be independent random variables taking values in $[0,1]$ and define $X := \frac{1}{n}\sum_{i \in [n]} X_i$. Then for any $\epsilon\in[0,1]$, we have
\[
\p\left[ |X - \E[X]|\geq \epsilon \right]\leq 2\exp\left(-n\epsilon^2/2 \right).
\]
\end{theorem}

\begin{theorem}[Freedman's Inequality, Theorem 1.6 in~\cite{Fre75}]
\label{thm:FreedmanInequality}
Consider a real-valued martingale difference sequence $\{X_t\}_{t\geq 0}$ such that $X_0=0$, and $\E[X_{t+1}|\mathcal{F}_t]=0$ for all $t$, where $\{\mathcal{F}_t\}_{t\geq 0}$ is the filtration defined by the martingale. Assume that the sequence is uniformly bounded, i.e., $|X_t|\leq M$ almost surely for all $t$. Now define the predictable quadratic variation process of the martingale to be $W_t=\sum_{j=1}^t \E[X_j^2|\mathcal{F}_{j-1}]$ for all $t\geq 1$. Then for all $\ell \geq 0$ and $\sigma^2>0$ and any stopping time $\tau$, we have
\[
\p\Big[ \Big|\sum_{j=0}^\tau X_j \Big|\geq \ell \wedge W_\tau \leq \sigma^2 \text{for some stopping time } \tau \Big] \leq 2\exp\Big(- \frac{\ell^2/2}{\sigma^2+M \ell/3} \Big).
\]
\end{theorem}

\subsection{The Doubling Technique}
\label{sec:appendix_doubling}

\subsubsection{Proof of Lemma~\ref{lm:larger_better}}
\label{sec:larger_better}
\newtheorem*{lm:larger_better}{Lemma~\ref{lm:larger_better}}
\begin{lm:larger_better}
For any $j\geq i\geq 1$ and any Algorithm $\ALGmeta$, one has 
\begin{align*}
    \E[\TC(\OPT(\M_i,k_i))]\geq \E[\TC(\OPT(\M_j,k_j))].
\end{align*}
Notice that the randomness is over an entire run of 
$\ALGgeneral$. 
\end{lm:larger_better}

\begin{proof}
For any phases~$i\geq 0$, we define:
\begin{itemize}
    \item $\omega_{i,\S}$:  the trajectory of \MS $\S$ traversed by $\ALGgeneral$ in the first $i$ phases.
    \item $\omega_i=\cup_{\S\in\mathbf{S}}\omega_{i,\S}$: the collection of the trajectories of all \MS systems traversed by $\ALGgeneral$ in the first $i$ phases.
\end{itemize}

At a first glance, $\OPT(\M_j,k_j)$ is a sub-tree of $\OPT(\M_i,k_i)$ and this inequality holds directly. But this is not true as $\OPT(\M_j,k_j)$ is required to start at the root $\mathbf{R}$.

In particular, we know that $\M_i$ and $k_i$ can be determined by $\omega_{i}$.
First, for all $j\geq i\geq 0$, we have that 
    $$
    \E[\TC(\OPT(\M_i,k_i))]=\E_{\omega_{j}}\left[\E[\TC(\OPT(\M_i,k_i)) \mid \omega_{j}]\right].
    $$
    Now, fixing any possible $\omega_{j}$, it suffices to prove that
    \begin{align}
    \label{eq:opti_and_optj}
        \E[\TC(\OPT(\M_i,k_i))\mid \omega_{j}]\geq \E[\TC(\OPT(\M_j,k_j))\mid \omega_{j}].
    \end{align}
    Note that we can represent the strategy $\OPT(\M_i,k_i)$
    as a decision tree. %and $\OPT(\M_j,k_j)$ will be a sub-tree of $\OPT(\M_i,k_i)$.
    No matter what $(\M_j,k_j)$ is for the $j$-th phase, we consider the algorithm $\ALG(\M_j,k_j)$ which uses the decision tree of $\OPT(\M_i,k_i)$, pretending not to know $\omega_{j}\setminus\omega_{i}$.
    The only difference is that we do not charge the movement cost of $\ALG(\M_j,k_j)$ for those transitions occurring at $\omega_{j}\setminus\omega_{i}$.
    
    Thus we know that
    \begin{align*}
        \E[\TC(\OPT(\M_i,k_i))\mid \omega_{j}]\geq  \E[\TC(\ALG(\M_j,k_j))\mid \omega_{j}]
        \geq  \E[\TC(\OPT(\M_j,k_j))\mid \omega_{j}],
    \end{align*}
    where the second inequality comes from the optimality of $\OPT(\M_j,k_j)$.
    By integrating over all possible $\omega_{j}$, we can complete the proof.
\end{proof}

\subsubsection{Proof of Lemma~\ref{lem:key_sub_process}}
\newtheorem*{lem:key_sub_process}{Lemma~\ref{lem:key_sub_process}}
\begin{lem:key_sub_process}
We are given some \SFindK (or \FindOne) instance $\M$, objective number of rewards $k$ and non-negative real number $B\in \mathbb{R}_{\geq 0}$. 
For any $B> c_1\E[\TC(\OPT(\M,k))]$, if we can get an algorithm $\ALGmeta$ with $\E[\ALGmeta]\leq c_2B$ and can succeed with probability more than $0.01$ where $c_1,c_2$ are some universal constant, then we can get an $O(1)$-approximation algorithm for \SFindK (or \FindOne).
\end{lem:key_sub_process}

\begin{proof}
Assume that $\min_{\S\in\mathbf{S}}d(\S,\mathbb{R})=1$.
If there is an algorithm $\ALGmeta$ satisfying the precondition in this lemma, we can put it into the general framework $\ALGgeneral$ and argue that $\ALGgeneral$ is an $O(1)$-approximation algorithm for \SFindK (or \FindOne).
%\footnote{We can assume we start at the root $\mathbf{R}$ at the beginning for each phase, the algorithm $\ALGgeneral$ costs another $O(1)\beta^i$ more for $i$-th phase.}
Let $\M,K$ be the original input in Algorithm~\ref{alg:alg_general} and for simplicity let $\OPT$ be the optimal strategy $\OPT(\M,K)$. 

We need some notations for any phases~$i\geq 0,j \geq 1$: 
\begin{itemize}[topsep=0cm,itemsep=0cm]

    \item $u_{j}(\omega_{i})$:  probability that $\ALGgeneral$ enters phase ${j}+1$ conditioning on $\omega_{i}$.
    %\item $u_{j}^*(\omega_{i})$: probability that the cost of \OPT is more than $\beta^{j}$, conditioning on $\omega_{i}$.
    \item $u_{j}$:  probability that $\ALGgeneral$ enters phase ${j}+1$.
    %\item $u_{j}^*$: probability that the cost of \OPT is more than $\beta^{j}$.
    %\item $\OPT(\M_i,k_i)$: the optimal strategy to the instance accepted by \ALGmeta in the $i$-th phase.

%    \item $k(\omega_{i-1})$: the remaining target of reward when entering phase~$i$.
%    \item $\ind_{k(\omega_{i-1}) > 0}$: indicator that $\ALGmeta$ visits vertices found in phase~$i$.
\end{itemize}

Notice that $u_{i-1}(\omega_{i-1})$ is the indicator variable that $\ALGgeneral$ enters phase~$i$. 
For $i$-th phase, as the budget of $\ALGmeta(\M_{i-1},k_{i-1},B_{i})$ is $B_{i}=c_2\beta^i$, the total cost of $\ALGgeneral$ is at most $2c_2\beta^i$ for $i$th phase.
Note that $\E[\TC(\OPT(\M_i,k_i))]$ is a random variable with randomness from $\omega_{i}$. %And let $\ETC(\OPT(\M_i,k_i))(\omega_{i-1})$ represent the realization of it with respect with $\omega_{i-1}$. 

Let $\ind_{\omega_{i-1}}$ be the indicator variable that $\E[\TC(\OPT(\M_{i},k_i))\mid \omega_{i-1}] \geq \frac{1}{c_1}\beta^i$.
If the precondition in the statement holds, we know that there exists some universal constant  $\beta>1$, for any phase~$i\geq 0$, and any possible $\omega_{i-1}$ such that the algorithm $\ALGgeneral$ satisfies
\begin{equation}
\label{eq:doubling_tech}
    u_i(\omega_{i-1}) \leq \frac{u_{i-1}(\omega_{i-1})}{\beta^2} + \ind_{\omega_{i-1}}.
\end{equation}

To make it more clear, we only need to consider the case when $u_{i-1}(\omega_{i-1})=1$ and $\ind_{\omega_{i-1}}=0$, as $u_{i-1}(\omega_{i-1})=0$ can imply $u_{i}(\omega_{i-1})=0$ directly and the inequality holds.
Besides, if $\ind_{\omega_{i-1}}=0$, then we know 
$u_i(\omega_{i-1})\geq 1-0.01=0.99$ by the precondition of the statement.
Setting $1/\beta^2\geq 0.99$ can handle this case.
%The last case is when $u_{i-1}(\omega_{i-1})=1$ and $ $

Now our objective is to show that $\E[\TC(\ALGgeneral)]\leq O(1)\cdot \E[\TC(\OPT)]$ by using Equation~\ref{eq:doubling_tech}.

As we have shown $\E[\TC(\OPT)\mid\omega_{i-1}]\geq \E[\TC(\OPT(\M_i,k_i))\mid\omega_{i-1}]$ in the proof of Lemma~\ref{lm:larger_better} for any possible $\omega_{i-1}$, then $\E[\TC(\OPT(\M_i,k_i))\mid\omega_{i-1}]\geq B$ means that $\E[\TC(\OPT)\mid\omega_{i-1}]\geq B$ for any real number $B$, which gives us:
 \begin{align*}
 \Pr[\E[\TC(\OPT(\M_i,k_i))\geq B] \leq & \Pr[ \E[\TC(\OPT)]\geq B]\\
 = & \ind_{ \E[\TC(\OPT)\geq B]},
 \end{align*}
where $\ind_{ \E[\TC(\OPT)\geq B]}$ is a deterministic indicator variable that $\E[\TC(\OPT)\geq B]$.
By taking expectation of both sides of Equation~\ref{eq:doubling_tech}, we have 
\begin{align*}
    u_i\leq u_{i-1}/\beta^2+ \Pr[ \E[\TC(\OPT(\M_i,k_i))\geq \frac{1}{c_1}\beta^i]] , \qquad \forall i\geq 0.
\end{align*}

Multiplying $\beta^i$ on both sides and taking summation of all $i$, we have 
\begin{align*}
    \sum_{i \geq 1} u_i \cdot \beta^i ~\leq & ~  u_0/\beta + 1/\beta \cdot \sum_{i \geq 1} u_i \cdot \beta^i +  \sum_{i\geq 1} \Pr[ \E[\TC(\OPT(\M_i,k_i))]\geq \frac{1}{c_1}\beta^i]\cdot \beta^i\\
    ~\leq &~ u_0/\beta + 1/\beta \cdot \sum_{i \geq 1} u_i \cdot \beta^i +  \sum_{i\geq 1} \Pr[ \E[\TC(\OPT)]\geq \frac{1}{c_1}\beta^i]\cdot \beta^i \\ 
    ~ = & ~  1/\beta + 1/\beta \cdot \sum_{i \geq 1} u_i \cdot \beta^i +  \sum_{i\geq 1} \ind_{ \E[\TC(\OPT)]\geq \frac{1}{c_1}\beta^i}\cdot \beta^i  \\
    ~ \leq & ~ 1/\beta + 1/\beta \cdot \sum_{i \geq 1} u_i \cdot \beta^i +  O(1)\cdot \E[\TC(\OPT)],
 \end{align*}
 
% where the second inequality is because that for any real number $B$, one has 
which implies that
\begin{align*}
(1-1/\beta) \cdot \sum_{i \geq 1} u_i \cdot \beta^i ~\leq~   1/\beta + O(1)\cdot\E[\TC(\OPT)].
\end{align*}

%Let $\EMC\OPT)$ and $\EMC\ALGgeneral)$ represent their expected cost, we have 
Besides, recall the assumption that  $\min_{\S\in\mathbb{S}}d(\S,\mathbb{R})=1$, one has
$
\E[\TC(\OPT)] \geq 1
$
and that
\[
\E[\TC(\ALGgeneral)]
\quad  \leq \quad  O(1)\cdot\sum_{i\geq 0}(u_i-u_{i+1})\cdot\beta^{i+1} \quad  = \quad  O(1)(\beta-1)\cdot\sum_{i\geq 0}u_i\cdot\beta^i.
\]
Combining these together, we get $\E[\TC(\ALGgeneral)] \leq O(1) \cdot \E[\TC(\OPT)]$.
\end{proof}
\section{Appendix for Section~\ref{sec:unit}}

\subsection{Algorithm.}
Algorithm~\ref{sec:alg_findone} is simply an instantiation of 
the doubling framework Algorithm~\ref{alg:alg_true_findone}.\\
\label{sec:alg_findone}
\begin{algorithm2e}[H]
\caption{Algorithm for \FindOne (for analysis purpose)}
\label{alg:alg_analysis_findone}
%{\bf  Input:} The instance metric $\M$\\
%{\bf Process:}\\
Set $\beta\in(1,2)$\;
Set $k_0=1$\;
 Set $\M_0=\M$\;
 Set $c_1=O(1)$\;
\For{phase $i=0,1,\cdots $}
{
    $(\M_i,k_i)\leftarrow \METAfindone(\M_i,c_1\beta^i)$ (Algorithm~\ref{alg:meta_findone})\;
    \If{$k_i\leq 0$}{\textbf{Break}}
}
\end{algorithm2e}

One has the following claim:
\begin{claim}
\label{cl:altanative_const}
%If Algorithm~\ref{alg:alg_analysis_findone} is an $O(1)$-approximation strategy for \FindOne, then so is Algorithm~\ref{alg:alg_true_findone}.
Algorithm~\ref{alg:alg_analysis_findone} has expected cost at least that of Algorithm~\ref{alg:alg_true_findone}. 
\end{claim}

\begin{proof}
Actually, Algorithm~\ref{alg:alg_analysis_findone} and Algorithm~\ref{alg:alg_true_findone} proceed in the same manner, but whilst Algorithm~\ref{alg:alg_true_findone} halts when finding a unit reward, Algorithm~\ref{alg:alg_analysis_findone} continues until the end of the phase. Algorithm~\ref{alg:alg_true_findone} only costs less.
\end{proof}
By the previous claim, it suffices to prove that
Algorithm~\ref{alg:alg_analysis_findone} achieves a constant 
approximation factor for \FindOne.

\subsection{Proof of Lemma~\ref{lm:good_proerty_large_T}}
\label{sec:prob_toss_coins}
\newtheorem*{lm:good_proerty_large_T}{Lemma~\ref{lm:good_proerty_large_T}}
\begin{lm:good_proerty_large_T}
Suppose $X_1,X_2,\cdots,X_n$ are a sequence of random variables taking values in $\{0,1\}$, and $\F_j=\sigma(X_1,\cdots,X_j)$ is the filtration defined by the sequence. Given any real number $\mu$, if $\sum_{j=1}^n\E[X_j\mid{\cal F}_{j-1}]\geq \mu$, then
\begin{align*}
    \Pr[\sum_{j=1}^nX_j\geq 1]\geq 1-e^{-3\mu/8}.
\end{align*}
\end{lm:good_proerty_large_T}
\begin{proof}
We construct a Martingale difference sequence (MDS) $Y_1,\cdots,Y_n$ to prove this lemma where $Y_j=X_j-\E[X_{j}\mid \F_{j-1}]$.
It is easy to check that $\{Y_j\}$ is a MDS:
First $\E[|Y_j|]\leq 1$.
Second, $\E[Y_j\mid \F_{j-1}]=\E[X_j\mid \F_{j-1} ]-\E[X_j\mid \F_{j-1} ]=0$.

We try to use Freedman's Inequality (Theorem~\ref{thm:FreedmanInequality}).
One has 
\begin{align*}
    W_t=\sum_{j=1}^t\E[Y_j^2\mid \F_{j-1}]\leq \sum_{j=1}^t\E[X_j^2\mid\F_{j-1}]\leq \sum_{j=1}^{t}\E[X_j\mid \F_{j-1}].
\end{align*}

Hence we have
\begin{align*}
    \Pr[ \sum_{j=1}^n X_j=0]= &\Pr\left[\sum_{j=1}^n Y_j=-\sum_{j=1}^n \E[X_{j}\mid \F_{j-1}]\right]\\
    \leq & \Pr\left[\sum_{j=1}^nY_j\leq -\mu\wedge W_n\leq \mu \right]\\
    \leq & \exp(-3\mu/8).
\end{align*}

\end{proof}

\iffalse
\begin{proof}
We can use Technical Lemma~\ref{lem:concentration} in the Appendix to get similar result, which constructs a sequence of martingale and applies the Freedman's Inequality (Theorem~\ref{thm:FreedmanInequality}).
Here we show another interesting way to prove it. 

%In our situation, the concrete distributions of $X_j$ is limited by the \MS given in the input. 
%We relax this limitation and try to argue that for a stronger, the probability of $\sum_{j=1}^{2^{7}B}X_j\geq 1$ is at least $1-e^{-2}$.

We can transfer the above problem into such a scenario: suppose we are tossing coins randomly with the hope to make no coins head. We can construct coins adaptively as we wish with only one constraint. More specifically, we can make $j$-th coin head with probability $p_j$ after observing the tossing results of the first $j-1$-th coins, and the two constraint is that
i)we can construct at most $2^{7}\cdot B$ coins; 
ii)and $\sum_{j=1}^{2^{7}\cdot B}p_j\geq \mu$.

It suffices to show that the upper bound of the objective (the probability that no coin heads up) in this scenario is no more than $e^{-2}$. 
Fortunately, it is easy to see that the best strategy is non-adaptive, as once there is one coin head, the game is over.
In the domain of non-adaptive algorithm, making $p_j=\frac{\mu}{2^{7}B}$ for each $j\in[2^{7}\cdot B]$ is optimal, which can be proved by induction and elementary calculation. The best objective is $(1-\frac{\mu}{2^{7}\cdot B})^{2^{7}\cdot B}\leq e^{-\mu}$. This completes the proof.
\end{proof}
\fi

\subsection{A Counter Example}
\label{sec:counter_example}
The simple index-based strategy is a constant factor approximation in unit metric space. 
However, there is a simple counter example with non-unit switching cost in which
\METAfindone performs arbitrarily worse than the optimal strategy.

Consider the following instance: We are given a metric space $\M=(\mathbf{S}\cup \{\mathbf{R}\} ,d)$. 
There are only two kinds of \MS. The first kind of \MS $\S=<\left\{s,t,x\right\},P,C,s,t>$, where $C_{s}=0, C_{x}=+\infty$, and $P_{s,t}=\epsilon,P_{s,x}=1-\epsilon,P_{x,t}=1$. The second kind of \MS $\S'$ has the similar structure, except that $P'_{s,t}=\epsilon/2$ and $P'_{s,x}=1-\epsilon/2$. The infinite set ${\bf S}=\left\{\S_1,\S_2,\cdots\right\}\cup\{\S'_1,\S'_2,\cdots\}$. 
%All of the second kind \MS locate at the same point on the metric space, and the switch-cost between any two of them are 0 (i.e. $d(\S'_{i},\S'_{j})=0$). As for the first kind of \MS, we have  $d(\S_i,\S_j)=d(\S_i,\S'_j)=1$.

Let $d(\mathbf{R},\S_k)=1$ for all $k$ and $d(\mathbf{\S_k},\mathbf{\S_j})=1$ for $k\neq j$.
Let $d(\S_j',\mathbf{R})=d(\S_j',\S_k)=\sum_{i=1}^{j}2^{-i+1}$ for all $k$.
As for the distances between the second kind of \MS, for $k<j$, we set $d(\S'_k,\S'_j)=\sum_{i=k-1}^{j-1}2^{-i}$.
 
 We are at $\mathbf{R}$ in the beginning and need to pay the switching cost to advance any one of \MS. The optimal strategy is obviously pay unit switching cost and continue playing the second kind of \MS until getting unit reward, and the total cost is no more than 2. 
But the greedy algorithm~\ref{alg:alg_true_findone} always chooses one of the first kind of \MS $\S_i$ to advance, with the expected cost $1/\epsilon$, which can be large arbitrarily.

\section{Appendix for Section~\ref{sec:general}}
\subsection{Estimate Order Statistics}
\label{sec:est_order_stat}

\begin{definition}[Permanent]
The permanent of an $n\times n$ matrix $A$ is defined by
\begin{align*}
    \per(A)=\sum_{\sigma}A_{i,\sigma(i)},
\end{align*}
where the sum is over all permutations $\sigma$ of $\{1,2,\cdots,n\}$.
\end{definition}

\begin{theorem}[Bapat-Beg theorem]
Let $X_1,\cdots,X_n$ be independent real valued random variables with cumulative distribution functions respectively $F_1(x),\cdots,F_n(x)$.
Write $X_{(1)},\cdots,X_{(n)}$ for the order statistics.
Then the joint probability distribution of the $n_1,\cdots,n_k$ order statistics (with $n_1<n_2<\cdots< n_k$ and $x_1<x_2<\cdots<x_k$) is
\begin{align*}
\begin{aligned}
F_{X_{\left(n_{1}\right)}, \ldots, X_{\left(n_{k}\right)}}\left(x_{1}, \ldots, x_{k}\right) &=\operatorname{Pr}\left(X_{\left(n_{1}\right)} \leq x_{1} \wedge X_{\left(n_{2}\right)} \leq x_{2} \wedge \cdots \wedge X_{\left(n_{k}\right)} \leq x_{k}\right) \\
&=\sum_{i_{k}=n_{k}}^{n} \cdots \sum_{i_{2}=n_{2}}^{i_{3}} \sum_{i_{1}=n_{1}}^{i_{2}} \frac{P_{i_{1}, \ldots, i_{k}}\left(x_{1}, \ldots, x_{k}\right)}{i_{1} !\left(i_{2}-i_{1}\right) ! \cdots\left(n-i_{k}\right) !},
\end{aligned}
\end{align*}
where 
\begin{align*}
    \begin{array}{l}
P_{i_{1}, \ldots, i_{k}}\left(x_{1}, \ldots, x_{k}\right)=\\
\per\left[\begin{array}{cccc}
F_{1}\left(x_{1}\right) \cdots F_{1}\left(x_{1}\right) & F_{1}\left(x_{2}\right)-F_{1}\left(x_{1}\right) \cdots F_{1}\left(x_{2}\right)-F_{1}\left(x_{1}\right) & \cdots & 1-F_{1}\left(x_{k}\right) \cdots 1-F_{1}\left(x_{k}\right) \\
F_{2}\left(x_{1}\right) \cdots F_{2}\left(x_{1}\right) & F_{2}\left(x_{2}\right)-F_{2}\left(x_{1}\right) \cdots F_{2}\left(x_{2}\right)-F_{2}\left(x_{1}\right) & \cdots & 1-F_{2}\left(x_{k}\right) \cdots 1-F_{1}\left(x_{k}\right) \\
\vdots & \vdots & \vdots \\
{\underbrace{F_{n}\left(x_{1}\right) \cdots F_{n}\left(x_{1}\right)}}_{i_{1}} & \underbrace{F_{n}\left(x_{2}\right)-F_{n}\left(x_{1}\right) \cdots F_{n}\left(x_{2}\right)-F_{n}\left(x_{1}\right)}_{i_{2}-i_{1}} & \cdots & \underbrace{1-F_{n}\left(x_{k}\right) \cdots 1-F_{n}\left(x_{k}\right)}_{n-i_{k}}
\end{array}\right]
\end{array}
\end{align*}
\end{theorem}

\begin{theorem}
There exists a fully polynomial randomized approximation scheme for the permanent of an arbitrary $n\times n $ matrix A with nonnegative entries.
\end{theorem}

%, and denote the (random) set by $F=\{\S\in \Pi\mid \D(\S)\leq \D_{[k]}(\S)\}$.
\begin{lemma}
For the set of distributions $\D^1,\D^2,\cdots,\D^n$ from which we can take i.i.d. samples, define $\D^{[K]}$ be the $K$-th order statistic.
For any real number $x\in \R$, we can estimate the value $\Pr[\D^{[k]}\leq x]$ within additive error $\epsilon$ with probability at least $1-1/\poly(n)$ in polynomial time.
\end{lemma}

\begin{proof}
$O(\log n)$ samples are enough to get an estimation of $\Pr[\D^{[k]}(\S)\leq x]$ within additive error $\epsilon=0.01$ with probability at least $1-1/\poly(n)$ by using Chernoff-Hoeffding Bound
(Theorem~\ref{thm:ChernoffBound}). 

This can also be replaced by the fully polynomial randomized approximation scheme (FPRAS) in \cite{JSV04} to estimate the permanent of $n-k$ matrices constructed according to Bapat-Beg theorem.
\end{proof}

\subsection{Computation of Index and CDF}
\label{appd:computation_cdf}

\begin{lemma}
\label{lm:cal_cdf}
We can calculate the cumulative distribution function of $\D(\S)$ efficiently.
\end{lemma}

\begin{proofof}{Lemma~\ref{lm:cal_cdf}}
Consider the \MS $\langle V,P,C,s,t\rangle$.
For simplicity, we denote the Gittins index of states in $V$ as $\gamma_1< \gamma_2< \cdots< \gamma_n$ where $|V|=n$ and $\gamma_1=0$ being the Gittins Index of the target state $t$.

For any $\gamma_i\leq B<\gamma_{i+1}$, we know that
\begin{align*}
    \Pr[\CSC(\S)\leq B]=\Pr[\CSC(\S)\leq \gamma_i].
\end{align*}

Let $U$ contain those states whose Gittins Index are larger than $\gamma_i$.
For state $v$, let $x_v$ be the probability that the non-quitting player can get no more than $\gamma_i$ units of profits under the rule of $\S^T$ conditioning on he is at state $v$ currently.
Then we know that
$x_v=\sum_{u\in N(v)}P_{v,u}x_u $, where $N(v)=\{u\in V\mid\gamma_u\leq \gamma_j\}$ if $v\notin U$; and $x_v=0$ if $v\in U$.

Then it involves solving a $i\times i$ system of equations, which can be done some standard techniques like LU decomposition in $O(i^3)$ time.
\end{proofof}

\subsection{Algorithm Framework}

We present our Algorithm Framework (Algorithm~\ref{alg:alg_SfindK}) here.

\begin{algorithm2e}[H]
\caption{Algorithm  $\ALGSfindK$}
\label{alg:alg_SfindK}
{\bf  Input:} The instance $\M$, objective number of rewards $K$\\
{\bf Process:}\\
 set $\beta\in (1,2)$\;
 set $k_0\leftarrow K$\;
 set $\M_0\leftarrow \M$\;
 \For{phase~$i=1,\cdots$}
 {
 $(\M_i,k_i)\leftarrow\METASfindK(\M,K,50000\beta^i)$\;
 \If{$k_i\leq 0$}{\textbf{Break}}
% Update $\M$ and $K$ according to the result\;
 %\If{$K<=0$}
 %{ {\bf Break}\;}
 }
\end{algorithm2e}

\subsection{Omitted Proof}
\label{sec:omitted proof}
\subsubsection{Proof of Claim~\ref{clm:relation_two_games}}
\newtheorem*{clm:relation_two_games}{Claim~\ref{clm:relation_two_games}}
\begin{clm:relation_two_games}
The following inequality holds:
$
    \E[\Ftotalcost(\P)]=\E[\SC(\FOPT)+\FMC(\FOPT)]\leq \E[\TC(\OPT)].
$
\end{clm:relation_two_games}

\begin{proof}
From another perspective, the player is playing a series of the teasing game $\S^T$ with switching cost. 
In the rule of $\S^T$, the player can get some amount of profits on the target state of $\S$, while in the rule $\SFindKF$, the player should pay the same amount of fair movement cost.

By Lemma~\ref{lm:fairness}, we know that $\S^T$ is a fair game.
In fact, if we exempt the switching cost, we can treat $\FMC(\OPT)$ as the summation of profits and $\MC(\OPT)$ as the movement cost of a series of fair game $\S^T$ for the strategy \OPT.
Because of the fairness, one has $\E[\FMC(\OPT)]\leq \E[\MC(\OPT)]$ and further
\begin{align}
\label{eq:first_Fcost}
    \E[\SC(\OPT)+\FMC(\OPT)]\leq \E[\TC(\OPT)].
\end{align}
Then by the optimality of $\FOPT$, we know
\begin{align}
\label{eq:fopt<opt}
    \E[\SC(\FOPT)+\FMC(\FOPT)]\leq \E[\SC(\OPT)+\FMC(\OPT)].
\end{align}
Combining Equation~\ref{eq:first_Fcost} and Equation~\ref{eq:fopt<opt}, we
complete the proof of the claim.
\end{proof}

\subsubsection{Proof of Claim~\ref{clm:games_reduction}}
\newtheorem*{clm:games_reduction}{Claim~\ref{clm:games_reduction}}
\begin{clm:games_reduction}
With the notations defined above, it holds that
\begin{align*}
    \E[\TSPCOST(\SCOPT)]=\E[\SC(\SCOPT)+\CSC(\SCOPT)]\leq \E[\TC(\OPT)].
\end{align*}
\end{clm:games_reduction}
Recall that we have defined the game \SFindKF and $\E[\FMC(\P)+\SC(\P)]$ is the expected fair movement cost plus the expected switching cost of any strategy $\P$ under the rule of \SFindKF. 
Now we let $\E[\CSC(\P)]$ and $\E[\SC(\P)]$ be the expected selected cost and expected switching cost of any strategy $\P$ under the rule of \StochKCost respectively. 
Let $\SCOPT$ be the optimal strategy to \StochKCost. 
Note that we use the same notation $\SC(\cdot)$ to represent the expected switching cost as these three games have the same rules about switching costs. 

Given any instance $\M$, let $\OPT$ and $\FOPT$ be the optimal strategies for this instance $\M$ under the rule of \SFindK and \SFindKF respectively, and let $\M_{SC}$ be the \StochKCost instance transferred from $\M$ by the above argument.
Now we show how to prove claim~\ref{clm:games_reduction}.

\begin{proof}
In fact, with the notation defined above, one has 
\begin{align*}
    \E[\SC(\SCOPT)+\CSC(\SCOPT)]\leq \E[\SC(\FOPT)+\FMC(\FOPT)]\leq \E[\TC(\OPT)].
\end{align*}

We only need to prove the first inequality as the second equality has been proved in Claim~\ref{clm:relation_two_games}. 

%\textcolor{green}{May need some assumptions such as there is no loop in the chain, or there is a limitation on the cost, or absorb the failure possibility for too large cost case.}

To see the first inequality, we can look at \StochKCost from another perspective.
Suppose we know the decision tree of $\FOPT$, and construct another strategy $\ALG^F$ based on the decision tree under a special rule equivalent to \StochKCost: whenever $\FOPT$ plays some system $\S_i$ for the first time, $\ALG^F$ can observe the whole sample path on $\S_i$ and know the largest Index of states on the sample path, and it can take one unit reward from $\S_i$ later by paying the corresponding largest Index. 
%Denote the expected movement cost of $\ALG^F$ by $\E[\CSC(\ALG^F)]$. 
Then we know that $\E[\SC(\ALG^F)+\CSC(\ALG^F)]= \E[\SC(\ALG^F)+\FMC(\ALG^F)]$,
%as $\ALG^F$ can get more information under the spacial rule.
as the distribution on fair movement cost of $\ALG^F$ is exactly the distribution on selected cost if we play \StochKCost by $\ALG^F$. 

Besides, we know $\E[\SC(\SCOPT)+\CSC(\SCOPT)]\leq\E[\SC(\ALG^F)+\CSC(\ALG^F)]$ by the optimality, which completes the proof of this claim.
\end{proof}

\section{Non-optimality of Indexing Strategies}
\label{sec:no_opt_index}
In this section, we extend the result in \cite{BS94} to show that even when the switching cost is a constant, there is no optimal indexing strategy for our problem in general metric.
More specifically, we define what is an index as follows:

\begin{definition}[\cite{BS94}]
\label{def:index}
An index in the presence of switching cost is any function $\gamma$ which specifies a value $\gamma(s_i,\S_i,I_i)$ for any \MS $\S_i$ where $s_i$ is the current state of \MS $\S_i$ and $I_i$ is the indicator that $\S_i$ is currently active.
\end{definition}

Involving the indicator is to capture the intuition that for two identical systems, the active one should be more attractive than the inactive one.
See more discussion in \cite{BS94}.

We say $\gamma$ is an optimal index in the presence of switching cost if it is optimal by always playing the system with the smallest index until any one reaches its target.

Now we define two kinds of \MS with specified notations.
First let $[x\delta_1+(1-x)\delta_0]=\{V,P,C,s,t\}$ where $V=\{s,t,v_0,v_1\}$, $P_{s,v_1}=x$, $P_{s,v_0}=1-x$, and $P_{x_i,t}=1$ for $i\in\{0,1\}$.
As for the movement cost, we let $C_s=c, C_{v_0}=0$ and $C_{v_1}=1$.
Second, let $[\delta_a]=\{\{s,t\},P,C,s,t\}$ where $P_{s,t}=1$ and $C_{s,t}=a$.
In the following, we assume the switching cost is $c$, same as $C_s$.

In this notation, we use $\gamma([x\delta_1+(1-x)\delta_{0}];I)$ and $\gamma(\delta_a;I)$ to denote the value of the optimal index on the \MS $[x\delta_1+(1-x)\delta_0]$ and $[\delta_a]$ respectively.
Recall that $I=1$ denotes that the system is currently active.

\begin{claim}[Similar to Claim 1 in \cite{BS94}]
Any index $\gamma$ that is optimal in the presence of switching cost $c$ must be a strict monotone transformation of an index $\hat{\gamma}$ which satisfies
\begin{gather}
    \hat{\gamma}([\delta_a];1)=a, \label{eq:active_delta_a}\\
    \hat{\gamma}([\delta_a];0)=a+c.\label{eq:inactive_delta_a}
\end{gather}
for any values of $a$ and $c$.
\end{claim}

Without loss generality, we assume the optimal index $\gamma$ satisfies equation~\ref{eq:active_delta_a} and equation~\ref{eq:inactive_delta_a}. Now we use the following claim to establish the impossibility of an optimal index:

\begin{claim}
There is no consistent way to define an index $\gamma$ on \MS of the form $[x\delta_1+(1-x)\delta_{0}]$ if the resulting strategy is to be invariably optimal. Consequently, an optimal index cannot exist.
\end{claim}

\begin{proof}
Consider such a game where there are two systems, $[\delta_a]$ and $[x\delta_1+(1-x)\delta_0]$ and suppose the optimal decision-player is on the first system $[\delta_a]$ and is indifferent from playing either one of the first step: he can either play $[\delta_a]$ or switch to $[x\delta_1+(1-x)\delta_0]$ and switches back if only if the state of $[x\delta_1+(1-x)\delta_0]$ moves to $v_1$.

In particular, we have $a=2c+x(a+c)$ and
elementary calculation shows that the value of $a$ should be $\frac{2+x}{1-x}c$ under the condition that $a+c\leq 1$ in this scenario.
Let $\mu(x,c):=\frac{2+x}{1-x}c$ and if $\mu(x,c)+c\leq 1$, we have 
\begin{align}
    \gamma\left(\left[x\delta_{1}+(1-x)\delta_{0} \right];0\right)=\gamma([\delta_{\mu(x,c)} ];1)=\mu(x,c)=\frac{2+x}{1-x}c.
\end{align}

Now consider a new game where suppose the optimal decision-player is on the second system and can either: switches to the first system directly; or plays $[x\delta_1+(1-x)\delta_0]$ for one step and switches to $[\delta_a]$ if the state moves to $v_1$.

Similarly, by simple calculation, let $v(x,c):=\frac{xc}{1-x}$ and under the condition that $v(x,c)+c\leq 1$ we have 
\begin{align}
    \gamma\left(\left[x\delta_1+(1-x)\delta_{0} \right];1\right)=\gamma([\delta_{v(x,c)} ];0)=v(x,c)+c=c/(1-x).
\end{align}

Hence we show that $\gamma\left(\left[x\delta_1+(1-x)\delta_{0} \right];0\right)=\frac{2+x}{1-x}c$ if $3c\leq 1-x$,
and $\gamma\left(\left[x\delta_1+(1-x)\delta_{0} \right];1\right)=c/(1-x)$ if $c\leq 1-x$ by the above two games.

Now consider the last game to show contradiction.
Suppose there are two systems $[x\delta_1+(1-x)\delta_0]$ and $[y\delta_1+(1-y)\delta_0]$ where $x\geq y$ but the player is current at system $[x\delta_1+(1-x)\delta_0]$. Suppose $3c\leq 1-x\leq 1-y$.
If we have
\begin{align*}
    \min\{c+x,c+x(2c+y)\}\leq \min\{2c+y,2c+y(2c+x)\},
\end{align*}
then the optimal strategy can play the first system for the first step and hence the constraint is $2x\leq 1+2y$ and we have
\begin{align*}
    \gamma([x\delta_1+(1-x)\delta_0];1)\leq \gamma([y\delta_1+(1-y)\delta_0];0).
\end{align*}
To conclude, the constraints are $y\leq x,3c+x\leq 1$ and $2x\leq 1+2y$. 

By setting $x=4/5$, $y=2/5$ and $c=1/100$, all the constraints are satisfied and one has
$\gamma([x\delta_1+(1-x)\delta_0];1)=\frac{1}{1-x}c=1/20>\gamma([y\delta_1+(1-y)\delta_0];0)=\frac{2+y}{1-y}c=1/25$, which is a contradiction completing the proof.
\end{proof}

\newpage
\addcontentsline{toc}{section}{References}
\bibliographystyle{alpha}
\bibliography{bib}

\end{document}